\documentclass[12pt, onecolumn]{IEEEtran}

\usepackage[T1]{fontenc}
\usepackage{amsmath}
\usepackage{amssymb}
\usepackage{microtype}
\usepackage{color}
\usepackage{bm}
\usepackage{ifpdf}

\ifpdf
\usepackage[pdftex]{graphicx}
\usepackage[update]{epstopdf}
\else
\usepackage{graphicx}
\fi

\renewcommand{\baselinestretch}{2}

\begin{document}

\newtheorem{conjecture}{Conjecture}
\newtheorem{rem}{Remark}
\newtheorem{insight}{Insight}
\newtheorem{question}{Question}
\newtheorem{proposition}{Proposition}
\newtheorem{cor}{Corollary}
\newtheorem{lem}{Lemma}
\newtheorem{assumption}{Assumption}
\newtheorem{theorem}{Theorem}
\newtheorem{example}{Example}
\newcommand{\FIM}{\mathbf{F}}
\newcommand{\MFIM}{\mathbf{F}_{\sf{M}}}
\newcommand{\param}{\beta}
\renewcommand{\Re}{\operatorname{Re}}
\renewcommand{\Im}{\operatorname{Im}}
\newcommand{\nuiv}{\mathbf{\widetilde{r}}}
\newcommand{\GCRB}{\mathsf{GCRB}(\phi_{21};\nuiv)}
\newcommand{\MCRB}{\mathsf{MCRB}(\phi_{21};\nuiv)}
\newcommand{\ACRB}{\mathsf{ACRB}(\phi_{21};\nuiv)}
\newcommand{\ACRBoneway}{\mathsf{ACRB}^{\mathsf{1-way}}(\phi_{21};\nuiv)}
\newcommand{\MCRBgen}{\mathsf{MCRB}(\paramreal_i; \{\paramreal_j, j\neq i\})}
\newcommand{\paramreal}{\widetilde{\param}}
\newcommand{\vparamreal}{\widetilde{\bm{\param}}}
\newcommand{\RCSO}{\mathbf{\Xi}} 
\newcommand{\rcso}{\bm{\xi}} 
\newcommand{\SCALE}{0.85}

\title{Carrier Frequency Offset Estimation for Two-Way Relaying: Optimal Preamble and Estimator Design}

\author{Chin Keong Ho, Patrick Ho Wang Fung, and~Sumei Sun%
\thanks{This paper was presented in part at the IEEE Workshop on Signal Processing Advances in Wireless Communications, 2010.}
\thanks{C. K. Ho and S. Sun are with the Institute for Infocomm Research, A*STAR, 1 Fusionopolis Way, \#21-01 Connexis, Singapore 138632 (e-mail: \{hock, sunsm\}@i2r.a-star.edu.sg).}
\thanks{P. H. W. Fung is with the Ngee Ann Polytechnic, Blk 7, \#04-01, 535 Clementi Rd., Singapore 599489. (e-mail: fhw2@np.edu.sg)}
}

\maketitle

\vspace{-1.5cm}

\begin{abstract}
We consider the problem of carrier frequency offset (CFO) estimation for a two-way relaying system based on the amplify-and-forward (AF) protocol. Our contributions are in designing an optimal preamble, and the corresponding estimator, to closely achieve the minimum Cramer-Rao bound (CRB) for the CFO. This optimality is asserted with respect to the novel class of preambles, referred to as the block-rotated preambles (BRPs). This class includes the periodic preamble that is used widely in practice, yet it provides an additional degree of design freedom via a block rotation angle. We first identify the catastrophic scenario of an arbitrarily large CRB when a conventional periodic preamble is used. We next resolve this problem by using a BRP with a non-zero block rotation angle. This angle creates, in effect, an artificial frequency offset that separates the desired relayed signal from the self-interference that is introduced in the AF protocol. With appropriate optimization, the CRB incurs only marginal loss from one-way relaying under practical channel conditions. To facilitate implementation, a specific  low-complexity class of estimators is examined, and conditions for the estimators to achieve the optimized CRB is established. Numerical results are given which corroborate with theoretical findings.
\end{abstract}

\vspace{-0.5cm}


\newpage

\renewcommand{\baselinestretch}{1.75}

\section{Introduction}

Two-way relaying is a spectrally efficient communication technique for two sources to exchange independent data \cite{Katti07, two-way_half_time_slot}. We consider two-way relaying based on analogue network coding, also known as the amplify-and-forward (AF) scheme with two transmission phases. In the first phase, both sources concurrently send their data to the relay, while in the second phase, the relay sends a scaled version of the received signal to both sources. Due to the sharing of spectral resources, the signal sent by one source is delivered not only to the other source, but also back to itself as self-interference through the relay. Assuming each source has knowledge of the two-way relay channel via channel estimation \cite{JiangGao10}, each source can subtract the self-interference and thus can recover the desired signal sent by the other source.

A potential application of two-way relaying for high data-rate transmissions over wireless channels is in multi-carrier systems, such as orthogonal frequency division multiplexing (OFDM) systems. It is well known that the presence of a carrier frequency offset (CFO) between the source and the destination can severely impair system performance. Hence, it is critical to consider the problem of CFO estimation in two-way relaying \cite{CFO_two-way, crucialCFO}, so that the detrimental effect can be mitigated at the receiver. Typically, a preamble is used to facilitate the estimation of the CFO.
In practical implementations, the CFO is estimated without any prior knowledge of the channel as it constitutes the first step in most communication systems, such as in the IEEE 802.11n Standard \cite{802.11n}, before any channel estimation is performed.

Given knowledge of the preambles used, CFO estimation in two-way relaying is fundamentally different from the classical problem of CFO estimation in point-to-point channels, or even in one-way-relaying. This is because in practice the channel is not known exactly and hence the removal of the self-interference is not straightforward.
Such self-interference corrupts the desired relayed signal and causes the CFO estimator to perform poorly. Despite some recent progress in this direction \cite{CFO_two-way, crucialCFO}, the fundamental reason for the loss in performance, if any, is not clear. In this regard, the Cramer-Rao bounds (CRB) serves as an important metric, since it gives the lowest possible variance for any unbiased estimator \cite{Kay}.

To understand the potential problem caused by self-interference, it is insightful to consider the following toy problem. Suppose we wish to estimate the frequencies of two tones received with unknown amplitudes and phases in the presence of additive white Gaussian noise (AWGN). Here the unknown amplitudes and phases represent the channel distortions.
The CRBs for the estimation of both frequencies turn out to be arbitrarily large as the frequencies approach each other \cite{RifeBoorstyn76}.
Compared to the two-way relaying case where one frequency (corresponding to the  CFO due to the desired relayed signal) is unknown while the other frequency (corresponding to the self-interference) is zero, this toy problem is a harder problem because both frequencies are unknown and to be estimated.
However, it captures the essence of the CFO estimation problem in two-way relaying, in that two different signals carried by unknown channels are present.
In fact, we shall see that both problems share the same fundamental limitation, namely that the CRB goes to infinity as the difference in the carrier frequencies approaches zero.
This motivates a re-design of the preamble used for CFO estimation, so as to remove this fundamental limitation.

Although the problem of preamble design in point-to-point channels for CFO estimation has been considered in the literature, e.g. \cite{Minn_cazac}, surprisingly there appears to have no such work in two-way relay systems.
In this paper, we will introduce a novel preamble design that in effect introduces an artificial frequency offset to remove the fundamental limitation that we have identified.

Our specific contributions are as follows. We consider the problem of preamble design and CFO estimation in a two-way relaying system, assuming a  time-invariant multipath wireless channel.
\begin{itemize}
\item We establish that reusing conventional periodic preambles at both sources, such as that used in the IEEE 802.11 standards \cite{802.11n}, can result in an unbounded CRB for the CFO estimator.
\item To overcome the above problem, we propose the novel class of {\em block-rotated preambles} (BRPs) for CFO estimation in two-way relaying. 
    The BRP includes the periodic preamble as a special case, yet provides an additional degree of design freedom via a block rotation angle.
    Intuitively, the block rotation angle introduces an artificial block-level frequency offset that enables the preambles from the two sources to be well separated in the frequency domain.
\item We obtain the CRB based on the class of BRP for the cases where the channel is either known a priori or not. To obtain a closed-form expression, we use an approximation of the CRB to optimize the BRP. Under this approximation, we show that the CRB for two-way relaying can approach the CRB for one-way relaying. i.e., the CRB where no self-interference is present.
\item To facilitate implementation, a specific class of estimators, based on linear filtering followed by conventional CFO estimation used in point-to-point transmission, is proposed. The necessary and sufficient condition for this class to achieve the one-way-relay CRB is derived.
\item Numerical results are obtained which corroborate with our analysis, and which illustrate the tightness of the approximations made for the BRP design.
\end{itemize}

This paper is organized as follows. The system model for the two-way relay is developed in Section~\ref{sec:System-Model}. The BRP is proposed in Section~\ref{sec:Proposed-Preamble-Structure}. The corresponding CRB is obtained in Section~\ref{sec:CRB} assuming some knowledge of the channel is available or none; an approximation of the CRB is also provided.
Next, the BRP is optimized in Section~\ref{sec:opt}. Section~\ref{sec:Effect-of-Filtering} proposes a low-complexity linear filter that does not suffer any loss in the CRB. Simulation results are given in Section~\ref{sec:Simulation-Results-And}, and finally conclusions are made in Section~\ref{sec:Conclusions}.

\emph{Notations}: Boldfaces in capital and small letters are reserved for denoting matrices and vectors respectively. All indices in matrices and vectors start from zero. The symbols $\otimes$, $\star$, $\mathbb{E}_{\mathcal{S}}(\cdot)$, $\mathbf{I}_{n}$ and $\mathbf{0}_{m\times n}$ represent convolution, Kronecker product operation, expectation function over the variables in the set $\mathcal{S}$, $n\times n$ identity matrix, and $m\times n$ null matrix, respectively.
In general, we collect the set of $N$ signals $\{r_n, n=1,\cdots, N\}$ as a column vector $\mathbf{r}=[r_1, \cdots, r_N]^T$.

\section{System Model\label{sec:System-Model}}

\begin{figure}
\centering{}\includegraphics[scale=0.8]{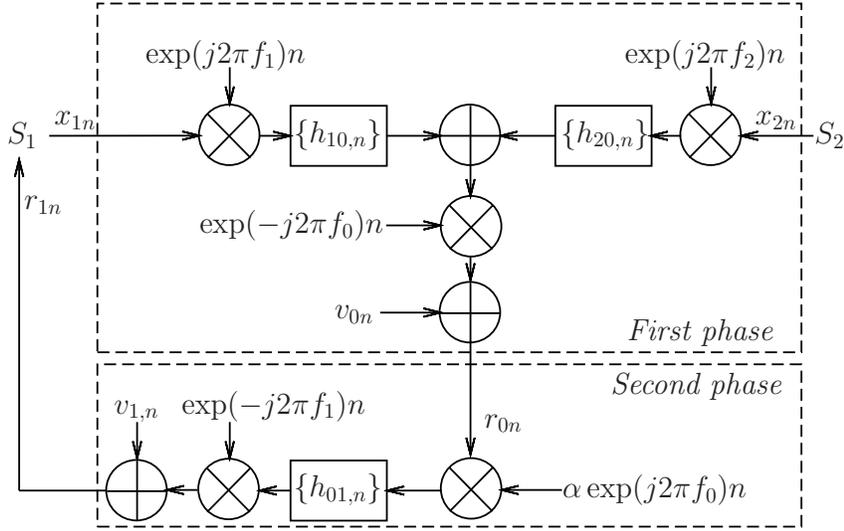}\caption{\label{fig:system model raw}Two-way relay system model from the perspective of source $S_1$.}
\end{figure}

We consider the two-way relay system consisting of one relay and two sources $S_{1}$ and $S_{2}$, referred to by the subscripts $0$, $1$ and $2$, respectively.
The relay and sources $S_{1}$ and $S_{2}$ transmit at carrier frequencies $f_{0}, f_1$ and $f_2$, respectively.
The carrier frequencies are set close to a pre-assigned value, but typically deviate slightly from one another due to hardware imperfections.
For the link from node $i$ to node $j$, we assume a $L$-tap frequency selective channel modeled by the finite-impulse response (FIR) $h_{ij,n}, n=0,\cdots, L-1$. The actual maximum number of the links can be less or equal to $L$ without loss of generality. For notational convenience, we collect the $L$ channel taps for the $ij$th link as $\mathbf{h}_{ij}=[h_{ij,0}, \cdots, h_{ij,L-1}]^T$.

We employ the amplify-and-forward protocol consisting of two phases. Figure \ref{fig:system model raw} illustrates the subsequent processing performed from the perspective of $S_1$.
In the first phase, both $S_{1}$ and $S_{2}$ concurrently transmit their packets $x_{1n} e^{j2\pi f_{1}n}$ and $x_{2n}e^{j2\pi f_{2}n}$ to the relay%
\footnote{If $S_{1}$ and $S_{2}$ do not transmit concurrently but with a small time difference, this delay is easily accommodated by introducing zeros in the first few samples of the channel FIR.}.
Unless otherwise specified, $n$ denotes the discrete time index that runs from $0$ to $N-1$.
Since we focus only on the CFO estimation problem, $\{x_{in}\}$ is taken to be the preamble of source $S_i, i=1,2$.
The relay down-converts the received signal to baseband by taking reference from its carrier frequency $f_0$. Hence the relay receives the baseband signal as
\begin{IEEEeqnarray}{rCl}\label{eqn:phase1rx}
r_{0n} & = & e^{-j2\pi f_{0}n}\left\{ \left(x_{1n}e^{j2\pi f_{1}n}\right)\otimes h_{10,n}+\left(x_{2n}e^{j2\pi f_{2}n}\right)\otimes h_{20,n}\right\} +v_{0n}, n=0,\cdots, N-1,
\end{IEEEeqnarray}
where $v_{0n}$ is the zero-mean AWGN with variance $\sigma^2_0$.
In the second phase, the relay scales $r_{0n}$ such that  $|\sum_{n=1}^N r_{0n}|^2/N$ equals the expected transmit power. We denote the scaling as $\alpha>0$.
Then the relay up-converts the baseband signal to its carrier frequency $f_0$, and broadcasts it to both sources.

Let us focus on the subsequent processing only at source $S_1$. The results at $S_{2}$ can be obtained similarly.
After down-converting the received signal from its carrier frequency $f_1$, source $S_1$ receives in baseband
\begin{IEEEeqnarray}{rCl}
r_{1n} &=& e^{-j2\pi f_{1}n}\left\{ \left(\alpha r_{0n}e^{j2\pi f_{0}n}\right)\otimes h_{01,n}\right\} +v_{1n}\label{eq:r1 version0}
\end{IEEEeqnarray}
where $v_{1n}$ is zero-mean AWGN with variance $\sigma^2_1$.

We have assumed that the relay performs a {\em digital} amplify-and-forward scheme, where the signals is scaled or amplified in the baseband. The above system model also holds if the relay performs an {\em analogue} amplify-and-forward scheme, where all processing is performed instead in the radio-frequency domain.

After some algebraic manipulations, we express \eqref{eq:r1 version0} as the following (more insightful)  signal model:
\begin{IEEEeqnarray}{rCl}
r_{1n} & =& r_{11,n}+e^{j2\pi(f_{2}-f_{1})n}r_{21,n}+u_{1n} \label{eq:r1 version01}
\end{IEEEeqnarray}
In \eqref{eq:r1 version01}, $r_{11,n}  \triangleq   h_{11,n}\otimes x_{1n}$ and
$r_{21,n}  \triangleq   h_{21,n}\otimes x_{2n}$ denote the equivalent received signals
with their equivalent channels given by
%
$ h_{11,n}  \triangleq \alpha e^{-j2\pi f_{1}n}  (h_{01,n} \otimes h_{10,n})$
and $h_{21,n}  \triangleq  \alpha  e^{-j2\pi f_{2}n}  (h_{01,n} \otimes h_{20,n})$, respectively,
while
\begin{IEEEeqnarray}{rCl}\label{eqn:noise}
u_{1n}  \triangleq   e^{j2\pi(f_{0}-f_{1})n} \left\{\left(\alpha  h_{01,n}e^{-j2\pi f_{0}n}\right) \otimes v_{0,n} \right\}+ v_{1n}.
\end{IEEEeqnarray}
Clearly, the equivalent noise $u_{1n}$ depends only on the relay-to-self channel tap $h_{01,n}$.

\section{The Class of Block-Rotated Preambles (BRPs)\label{sec:Proposed-Preamble-Structure}}
Our objective in this paper is to design the preambles $\{x_{1n}\}$
and $\{x_{2n}\}$, so that the CFO $(f_{2}-f_{1})$ is estimated as accurately as possible.
We restrict our study to the class of the block rotated preambles (BRPs) preambles, which is proposed in this section.
As we shall see, the BRP overcomes a significant problem when periodic preambles are used, and when optimized, the BRP approaches the ideal performance where no self-interference is present.

\subsection{Definition}

A BRP $\mathbf{x}=[x_0,\cdots, x_{(M+1)L-1}]^T$ of length $(M+1)L$ is uniquely defined by a {\em basis block} $\mathbf{b} = [b_0, b_2, \cdots, b_{L-1}]^T\in \mathbb{C}^{L\times 1}$ and a {\em block rotation angle} $\theta\in [0, 2\pi)$, according to
\begin{IEEEeqnarray}{rCl}\label{eqn:BRP}
\mathbf{x} = [\mathbf{b}^T, \rho \mathbf{b}^T, \rho^2\mathbf{b}^T, \cdots \rho^{M} \mathbf{b}^T]^T
\end{IEEEeqnarray}
where $\rho\triangleq \exp(j \theta)$.
If $\theta=0$ and so $\rho=1$, the BRP becomes a periodic preamble, which is used widely in conventional point-to-point communication systems, e.g. in \cite{802.11n}. Hence, we may treat the phase rotation $\theta$ as being applied to the periodic preamble on a {\em block level}, i.e., the phase remains constant for every block of $L$ samples then increments by $\theta$ for the next block.

\subsection{Equivalent System Model}

Henceforth, we assume source $S_i$ uses the BRP $\mathbf{x}_i$ with basis block $\mathbf{b}_i$ and block rotation angle $\theta_i$, $i=1,2$.
In practice, a time synchronization algorithm is used to estimate the arrival of the preambles at the receiver. Any timing error does not destroy the block rotation property of the BRP as given in \eqref{eqn:BRP}.
Hence, without loss of generality we assume perfect knowledge on  the arrival of the preambles. To remove any possible inter-symbol interference from other transmissions, we follow the common practice of discarding the first $L$ samples that constitute the first block of preambles.
Thus, we have $N=M L$ samples of received signal left. For notational convenience, we reset the time index $n$ in \eqref{eq:r1 version01} and \eqref{eqn:noise} to start from the second basis block.
From \eqref{eq:r1 version01}, after some algebraic manipulations we obtain the received signal at $S_1$ as
\begin{IEEEeqnarray}{rCl}
\mathbf{r}_{1} & = &
\underbrace{\left[\begin{array}{c}
\rho_{21}\mathbf{I}_{L}\\
\rho_{21}^{2}\mathbf{I}_{L}\\
\vdots\\
\rho_{21}^{M}\mathbf{I}_{L}
\end{array}\right]}_{\mathbf{G}_{21}}
\mathbf{r}_{21}+
\underbrace{\left[\begin{array}{c}
\rho_{11}\mathbf{I}_{L}\\
\rho_{11}^{2}\mathbf{I}_{L}\\
\vdots\\
\rho_{11}^{M}\mathbf{I}_{L}
\end{array}\right]}_{\mathbf{G}_{11}}\mathbf{r}_{11}+
\mathbf{u}_{1},\label{eq:discrete model at S1}
\end{IEEEeqnarray}
where $\rho_{21} \triangleq e^{j\phi_{21}}, \phi_{21}=2\pi(f_{2}-f_{1})L+\theta_{2}, \rho_{11} \triangleq e^{j\phi_{11}}, \phi_{11}=\theta_{1}$. We recall that $N=ML$ and the vector and matrix dimensions are $\mathbf{r}_1, \mathbf{u}_1\in \mathbb{C}^{N\times 1}$, $\mathbf{G}_{11}, \mathbf{G}_{21}\in \mathbb{C}^{N\times L}$, while $\mathbf{r}_{11}, \mathbf{r}_{21}\in \mathbb{C}^{L\times 1}.$

We interpret the terms in the system model \eqref{eq:discrete model at S1} as follows. The first term $\mathbf{G}_{21}\mathbf{r}_{21}$ is the signal vector that contains useful information of the CFO $(f_2-f_1)$ via $\mathbf{G}_{21}$ (which depends on $\phi_{21}$), while $\mathbf{r}_{21}$ is a channel-related nuisance parameter that depends on $\mathbf{h}_{01}, \mathbf{h}_{20}$.
The second term $\mathbf{G}_{11}\mathbf{r}_{11}$ is the self-interference vector at direct current that can potentially interfere with the CFO estimation, where $\mathbf{r}_{11}$ is another channel-related nuisance parameter that depends on $\mathbf{h}_{01}, \mathbf{h}_{10}$.
Subsequent results will support and further clarify the above intuitive view.
Finally, the last term  corresponds to the additive coloured Gaussian noise in \eqref{eqn:noise}, which can be expressed as%
\footnote{We make the dependence of the channel $\mathbf{h}_{01}$ explicit as this leads to the key difference between the two CRBs that we will introduce later.}
\begin{IEEEeqnarray}{rCl}\label{eqn:noisevec}
\mathbf{u}_1 = \mathbf{K}(\mathbf{h}_{01})\mathbf{v}_0 +  \mathbf{v}_1  
\end{IEEEeqnarray}
where 
$\mathbf{K}(\mathbf{h}_{01})\in\mathbb{C}^{N\times{N+L-1}}$ has the $i$th row as $[\mathbf{0}_{1\times (i-1)}, \mathbf{h}_{01}, \mathbf{0}_{1\times (N-i)}]$,
while $\mathbf{v}_0=[v_{0,-L+1}, \cdots, v_{0,N}]^T$ and $\mathbf{v}_1=[v_{1,1}, v_{1,2}, \cdots, v_{1,N}]^T$ are AWGN vectors.
In \eqref{eqn:noisevec}, without loss of generality we have discarded the phase rotation $e^{j2\pi(f_0-f_1)n}$ in \eqref{eqn:noise}, because all random distributions are assumed to be circularly symmetric.
Given $\mathbf{h}_{01}$, the vector $\mathbf{u}_1$ is Gaussian distributed with mean zero and covariance matrix
\begin{IEEEeqnarray}{rCl}\label{eqn:noisecov}
\mathbb{E}[\mathbf{u}_1\mathbf{u}_1^H]\triangleq \mathbf{R}(\mathbf{h}_{01})=  \mathbf{K}(\mathbf{h}_{01})  \mathbf{K}^H(\mathbf{h}_{01})+\mathbf{I}
\end{IEEEeqnarray}
where without loss of generality, we assume noise variances to be one, i.e., $\sigma^2_0=\sigma^2_1=1$.
We note that if $\mathbf{h}_{01}$ is random, then $\mathbf{u}_1$ is no longer Gaussian distributed in general.

\begin{rem}\label{rem:one-way}
If source $S_1$ does not transmit in the first phase, i.e., $x_{1,n}=0$ for all $n$, then we obtain a one-way relaying system. The system model for one-way relaying is thus given by \eqref{eq:discrete model at S1} with $\mathbf{r}_{11}=\mathbf{0}_{L\times 1}$.
\end{rem}

For subsequent derivations, it is convenient to re-write the {\em complex-valued} system model in \eqref{eq:discrete model at S1} as the {\em real-valued} system model:
\begin{IEEEeqnarray}{rCl}\label{eq:discrete model at S1_real}
\mathbf{r}=\mathbf{A}\nuiv+\mathbf{n}
\end{IEEEeqnarray}
where
$\mathbf{r} \triangleq \left[\Re(\mathbf{r}^T_{1}), \Im(\mathbf{r}^T_{1})\right]^T$,
$\mathbf{n} \triangleq \left[\Re(\mathbf{u}^T_{1}), \Im(\mathbf{u}^T_{1})\right]^T$,
$\mathbf{A} \triangleq
\left[\begin{array}{cr}
\Re\left(\left[
\mathbf{G}_{11},  \mathbf{G}_{21}\right]\right) & -\Im\left(\left[
\mathbf{G}_{11},  \mathbf{G}_{21}\right]\right)\\
\Im\left(\left[
\mathbf{G}_{11}, \mathbf{G}_{21}\right]\right) & \Re\left(\left[
\mathbf{G}_{11}, \mathbf{G}_{21}\right]\right)
\end{array}\right],$
and
$\nuiv  \triangleq  \left[\Re\left(\left[\mathbf{r}^T_{11}, \mathbf{r}^T_{21} \right]\right), \Im\left(\left[
\mathbf{r}^T_{11}, \mathbf{r}^T_{21} \right]\right)\right]^T$.

\section{Cramer-Rao Bound (CRB) for Preamble Design}\label{sec:CRB}

The CRB gives the fundamental limit of the variance of any unbiased estimator \cite{Kay}.
We focus on the CRB  for the CFO estimation only at source $S_1$; similar results hold for the CFO estimation at source $S_2$.
For simplicity, we consider the CRB of $\phi_{21}=2\pi(f_{2}-f_{1})L+\theta_{2}$, instead of the CFO $(f_2-f_1)$.
The CFO is related to $\phi_{21}$ by a linear transformation, and hence both CRBs are related simply by a linear transformation \cite{Kay}.

\subsection{General Approach}\label{sec:CRB:approach}

To obtain the true CRB, so-called to distinguish from the CRBs to be introduced, we have to take into account $\nuiv$ which is treated as nuisance parameters.
A closed-form expression for the true CRB however appears to be intractable.
Since our aim is to design preambles, it is more useful to have closed-form expressions based on lower bounds, or under suitable tight approximations, of the true CRB.
To this end, we establish two lower bounds of the true CRB, namely, 
the {\em genie-aided CRB} (GCRB) assuming the channel $\mathbf{h}_{01}$ is known in Section~\ref{sec:fullCSICRB}, and
the {\em modified CRB} (MCRB) \cite{GiniRegMengali} assuming that $\mathbf{h}_{01}$ is not known in Section~\ref{sec:noCSICRB}.
We also obtain the {\em approximate CRB} (ACRB) in closed-form which serve as a good approximation for the MCRB.


%


\subsection{Genie-Aided CRB (with Perfect Knowledge of $\mathbf{h}_{01}$)}\label{sec:fullCSICRB}

Theorem~\ref{thm:The-CRB-for} states the GCRB for the estimation of $\phi$ assuming knowledge of $\mathbf{h}_{01}$ is available.
Although the only desired parameter of interest is $\phi_{21}$,  the nuisance parameters $\nuiv$ are also jointly estimated in the derivations for the GCRB.
The GCRB (of $\phi_{21}$) is thus derived for a given parameter set $\bm{\param}\triangleq [\phi_{21},\nuiv^T]^{T}$, and thus denoted explicitly as $\GCRB$.



%

%
%
%
%
%
\begin{theorem}\label{thm:The-CRB-for}
Assume that $S_{i}$ transmits a BRP with basis block $\mathbf{b}_i$ and block rotation angle $\theta_i$, where $i=1,2$, and $M\geq 3$ basis blocks are used%
\footnote{See Corollary~\ref{cor:infCRB2} later which states that the GCRB for $M=2$ is not well defined.}.
Given the received signal $\mathbf{r}_{1}$ in (\ref{eq:discrete model at S1}) and the CSI $\mathbf{h}_{01}$, the GCRB of $\phi$ at source $S_1$ is 
\begin{IEEEeqnarray}{rCl}
\GCRB
=\left( {2\mathbf{r}_{21}^{H}\mathbf{G}_{21}^{H}\mathbf{T}\mathbf{\Phi}_{1}\mathbf{T}\mathbf{G}_{21}\mathbf{r}_{21}-2\mathbf{r}_{21}^{H}\mathbf{\Psi}_{1}\mathbf{r}_{21}}
\right)^{-1}
\label{eq:CRB complete}
\end{IEEEeqnarray}
where 
$\mathbf{T}\triangleq\mbox{diag}(0,1,2,\ldots M-1)\otimes\mathbf{I}_{L}$
and
\begin{IEEEeqnarray}{rCl}\label{eq:one-way relay condition0}
\mathbf{\Phi}_{1}\triangleq\mathbf{R}^{-1}-\mathbf{R}^{-1}\mathbf{G}_{21}\left(\mathbf{G}_{21}^{H}\mathbf{R}^{-1}\mathbf{G}_{21}\right)^{-1}\mathbf{G}_{21}^{H}\mathbf{R}^{-1},\\
\mathbf{\Psi}_{1}\triangleq\mathbf{G}_{21}^{H}\mathbf{T}\mathbf{\Phi}_{1}\mathbf{G}_{11}\left(\mathbf{G}_{11}^{H}\mathbf{\Phi}_{1}\mathbf{G}_{11}\right)^{-1}\mathbf{G}_{11}^{H}\mathbf{\Phi}_{1}\mathbf{T}\mathbf{G}_{21}.
\label{eq:one-way relay condition}
\end{IEEEeqnarray}
Moreover, $\mathbf{\Phi}_{1}$ satisfies $\mathbf{\Phi}_{1} \mathbf{G}_{21}=\mathbf{0}_{NL\times L}.$
\end{theorem}
\begin{proof}
See Appendix~\ref{app:Proof-of-Theorem 1}.
\end{proof}




\begin{rem}\label{rem:boundedcond}
In the GCRB, the nuisance parameters are treated as deterministic.
The true CRB, however, treats $\nuiv$ as random.
A lower bound of the true CRB is given by the extended Miller-and-Chang bound (EMCB) $\mathbb{E}_{\nuiv}[\GCRB]$, i.e., the expectation of the GCRB over the nuisance parameters $\nuiv$ \cite{GiniReggiannini}.
%
%
The EMCB is obtained assuming the channel $\mathbf{h}_{01}$ is known.
This bound is also a lower bound for the true CRB assuming $\mathbf{h}_{01}$ is not known, since this knowledge can always be discarded even if available to give the same estimator performance. 
\end{rem}

From Remark~\ref{rem:boundedcond}, Theorem~\ref{thm:The-CRB-for} provides a lower bound for the true CRB, whether $\mathbf{h}_{01}$ is known or not. 
Next, Theorem~\ref{cor:infCRB} states the necessary condition for the true CRB to be bounded.


\begin{theorem}
\label{cor:infCRB}
Assume the same conditions as in Theorem~\ref{thm:The-CRB-for} with $S_{1}$ and $S_{2}$ transmitting periodic
(but possibly different) preambles, i.e., $\theta_{1}=\theta_{2}=0$.
Then the GCRB, and also the true CRB whether the channel $\mathbf{h}_{01}$ is known or not, are unbounded as  $(f_{2}-f_{1})\rightarrow 0$ for any $\mathbf{h}_{01}$.
\end{theorem}
\begin{proof}
Appendix~\ref{app:infCRB} proves that $\GCRB\rightarrow\infty$ as $(f_{2}- f_{1})\rightarrow 0$ for any $\nuiv$.
By Remark~\ref{rem:boundedcond}, the true CRB is also unbounded asymptotically, whether the channel $\mathbf{h}_{01}$ is known or not,
\end{proof}

Theorem~\ref{cor:infCRB} shows that using periodic preambles (even different ones) at both sources leads to an unbounded true CRB if the carrier frequencies of these two sources are the same.
This result suggests that the problem of CFO estimation in two-way relay systems is similar in nature to the problem where two carrier frequencies are present and their values have to be estimated; in the latter problem, the CRBs of estimating the two carrier frequencies is arbitrarily large as the frequencies approach each other \cite{RifeBoorstyn76}.

In practice, the CFO $(f_{2}-f_{1})$ approaches zero by design but may not be exactly zero. Since the GCRB is a continuous function of the CFO, the CRB will still be large if the CFO is small. Thus, reusing conventional periodic preambles at both sources can lead to the potentially catastrophic scenario where the CFO effectively cannot be estimated, as confirmed also by numerical results in Section~\ref{sec:Simulation-Results-And}.

Theorem~\ref{thm:The-CRB-for} assumes that $M\geq 3$ basis blocks are used. Corollary~\ref{cor:infCRB2} shows that the GCRB is not well defined if $M=2$.

\begin{cor}\label{cor:infCRB2}
Assume the same conditions as in Theorem~\ref{thm:The-CRB-for}.
Then the GCRB is undefined if $M=2$ for any $\mathbf{h}_{01}$.
\end{cor}
\begin{proof}
See Appendix~\ref{app:infCRB2}
\end{proof}

Corollary~\ref{cor:infCRB2} suggests that the minimum number of basis blocks to be used is three. We note that Corollary~\ref{cor:infCRB2} holds for any BRP, including (conventional) periodic preambles. This is somewhat surprising, since for point-to-point and even one-way relay systems, two blocks of periodic preambles are sufficient for CFO estimation, see for example \cite{MM}.
Intuitively,
this is because the degrees of freedom are insufficient when $M=2$. Specifically, each of the two carrier frequencies is corrupted by a complex-valued attenuation, which result in a total of six (desired or nuisance) real-valued parameters. Consider the extreme case of a flat-fading channel and $L=1$ symbol is present in each basis block. Then $M=2$ basis blocks give only two complex-valued received signals, or four real-valued received signals, which are insufficient to estimate all the six parameters. On the other hand, $M=3$ basis blocks give just sufficient number of received signals to estimate all six parameters.

In view of Corollary~\ref{cor:infCRB2}, we assume henceforth that $M\geq 3$. Since $N=ML$ and $L\geq 1$, we have $N\geq 3$.

\subsection{Modified CRB (without Knowledge of $\mathbf{h}_{01}$)}\label{sec:noCSICRB}

The GCRB in (\ref{eq:CRB complete}) assumes knowledge of the channel $\mathbf{h}_{01}$, which  provides some insights to the true CRB.
In this section, we assume that $\mathbf{h}_{01}$ is not known but random with some known distribution, which is a more reasonable assumption in practice.
Similar to Section~\ref{sec:fullCSICRB}, we assume a given (deterministic) parameter set $\bm{\param}$ comprising the parameter of interest $\phi_{21}$ and the nuisance parameters $\nuiv$.
We shall employ the MCRB \cite{GiniRegMengali} to give a lower bound for the true CRB. 

Before specializing to our system, let us consider the following general real-valued system model:
\begin{IEEEeqnarray}{rCl}\label{eqn:signalmodel_corrnoise}
\mathbf{r}=\mathbf{z}(\vparamreal) + \mathbf{H} \bm{\epsilon}
\end{IEEEeqnarray}
where $\mathbf{r}, \mathbf{z}$ and $\bm{\epsilon}$ are of length $n$, and $\mathbf{H}$ is a square matrix.
In \eqref{eqn:signalmodel_corrnoise}, $\mathbf{r}$ is the received signal, $\mathbf{z}(\vparamreal)=[{z}_1(\vparamreal), \cdots, {z}_N(\vparamreal)]$ depends only on the length-$p$ parameter vector $\vparamreal=[\paramreal_1, \cdots, \paramreal_p]$, while $\mathbf{H}$ and $\bm{\epsilon}$ are random with known distributions.

\begin{lem}\label{lem:1}
Given $\mathbf{r}$ in \eqref{eqn:signalmodel_corrnoise}, a lower bound for the variance of any unbiased estimator for $\paramreal_i$, where $i=1,\cdots,p$, is given by the MCRB
\begin{IEEEeqnarray}{rCl}\label{eqn:MCRB}
\MCRBgen= \left[\MFIM^{-1}\right]_{ii}
\end{IEEEeqnarray}
where $\MFIM=\mathbb{E}_{\mathbf{r}}
\left[\dfrac{\partial g(\vparamreal)}{\partial\vparamreal}
\dfrac{\partial g(\vparamreal)}{\partial\vparamreal^T}\right] $
is the modified FIM, $g(\vparamreal)= \log f_{\mathbf{r}| \mathbf{H}}(\mathbf{r};\vparamreal|\mathbf{H})$ is the conditional log-likelihood function, and $[\mathbf{X}]_{ii}$ denotes the $i$th diagonal element of the matrix $\mathbf{X}$.
If the following assumptions hold:
\begin{itemize}
\item[$A1$:] $\mathbf{H}$ is full rank with probability one;
\item[$A2$:] the elements of $\bm{\epsilon}$ are i.i.d. (not necessarily Gaussian distributed);
\item[$A3$:] $\bm{\epsilon}$ and $\mathbf{H}$ are independent of each other, and also both
are independent of $\vparamreal$,
\end{itemize}
then the modified FIM simplifies as
%
\begin{IEEEeqnarray}{rCl}\label{eqn:FIM2}
\MFIM(\vparamreal) 
&=& \gamma  {\mathbf{Z}'}(\vparamreal) \mathbf{\Gamma} {\mathbf{Z}'}^T(\vparamreal)
\label{eqn:FIM2b}
\end{IEEEeqnarray}
where $\mathbf{Z}'$ is a $p$-by-$N$ matrix with the $(i,j)$th element as $\partial z_j({\bm \paramreal})/\partial \paramreal_i$.
In \eqref{eqn:FIM2b}, $\gamma$ is a scalar that depends only on the distribution of $\bm{\epsilon}$, while $\mathbf{\Gamma}\triangleq \mathbb{E}_{\mathbf{H}}\left[\left(\mathbf{H} \mathbf{H}^H\right)^{-1}\right]$ depends only on the distribution of $\mathbf{H}$.
\end{lem}
\begin{proof}
Taking $\mathbf{H}$ as the nuisance parameter, and applying the results in \cite{GiniRegMengali}, the MCRB  \eqref{eqn:MCRB} follows immediately as a lower bound to the variance of the estimator of $\paramreal_i$.
The simplification of the modified FIM \eqref{eqn:MCRB} to \eqref{eqn:FIM2} under assumptions $A1$-$A3$ is given in Appendix~\ref{app:FIM2}.
\end{proof}

We now apply Lemma~\ref{lem:1} to obtain a lower bound for the estimation of $\phi_{21}$ given $\mathbf{r}_1$ in \eqref{eq:discrete model at S1_real}.
Theorem~\ref{thm:CRBfinal} states the MCRB in general which is then expressed in closed-form for a special case.


\begin{theorem}\label{thm:CRBfinal}
Assume that $S_{i}$ transmits a BRP with basis block $\mathbf{b}_i$ and block rotation angle $\theta_i$, where $i=1,2$, for $M\geq 3$.
Then a lower bound for the variance of any unbiased estimator for $\phi_{21}$ is given by the MCRB
\begin{IEEEeqnarray}{rCl}\label{eqn:MCRBfreq}
\MCRB= \left[\MFIM^{-1}\right]_{ii}
\end{IEEEeqnarray}
where the modified FIM is given by \eqref{eqn:FIM2} with the general system model in \eqref{eqn:signalmodel_corrnoise} specialized to the system model in \eqref{eq:discrete model at S1_real}. Specifically, we use the substitutions
$\widetilde{\bm{\param}}=[\phi_{21},\nuiv^T]^T$ and
$\mathbf{z}(\bm{\paramreal})=\mathbf{A}\nuiv$ as in \eqref{eq:discrete model at S1_real},
which gives%
\footnote{From \eqref{eqn:noisecov}, $\mathbf{R}^{-1}(\mathbf{h}_{01})$ is always invertible. Hence, $\bm{\Gamma}$ exists.}
$\bm{\Gamma}=\mathbb{E}_{\mathbf{h}_{01}}[\mathbf{R}^{-1}(\mathbf{h}_{01})]$.
If $\bm{\Gamma}=k\mathbf{I}$ for some constant $k>0$, 
then the MCRB can be expressed in closed-form as%
\footnote{This special MCRB shall be used to denote an approximation of the MCRB later, hence we use the acronym ACRB.}
%
\begin{IEEEeqnarray}{rCl}\label{eqn:CRB-noCSI}
\ACRB=
\dfrac{6 k}{N[N^{2}-1-3\lambda(\phi_{21}-\phi_{11})](\mathbf{r}_{21}^{H}\mathbf{r}_{21})}
\end{IEEEeqnarray}
where the  {\em degradation function} $\lambda(x)$ is given by
\begin{IEEEeqnarray}{rCl}
\lambda(x)\triangleq\dfrac{\left[M\cos\left(Mx/2\right)
\sin\left(x/2\right)-\sin\left(Mx/2\right)
\cos\left(x/2\right)\right]^{2}}{\sin^{2}\left(x/2\right)\left[M^{2}\sin^{2}\left(x/2\right)-\sin^{2}\left(Mx/2\right)\right]}.
\label{eq:definition of lambda}
\end{IEEEeqnarray}
Clearly, $\ACRB$ strictly increases as $\lambda$ increases.
\end{theorem}
\begin{proof}
See Appendix~\ref{app:CRBfinal}.
\end{proof}

\begin{figure}
\centering{}\includegraphics[scale=\SCALE]{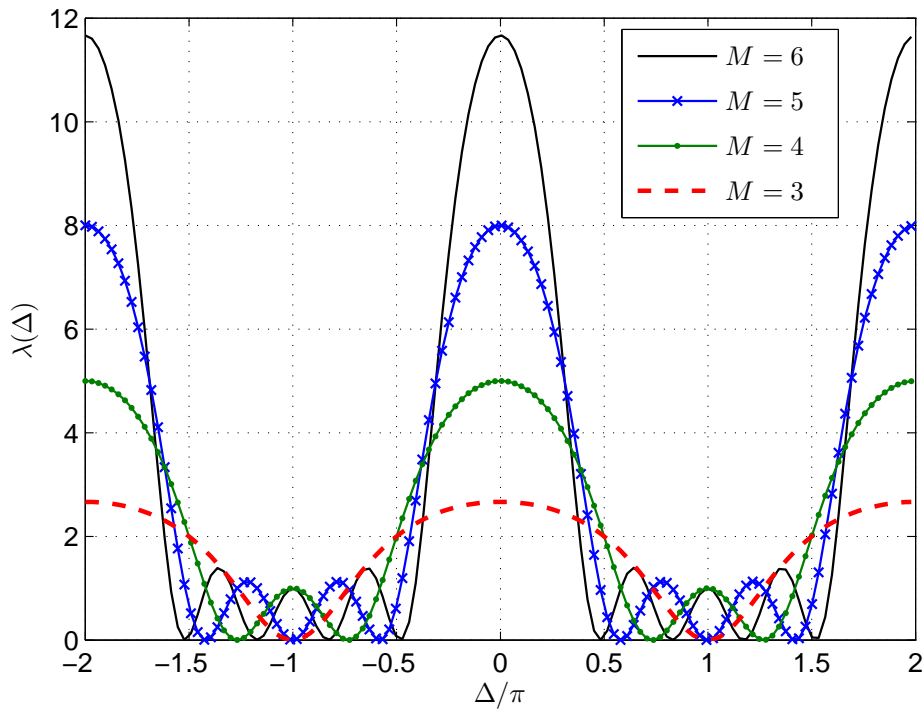}\caption{Graph of degradation function $\lambda(x)$ for $M=3,\cdots, 6$. \label{fig:degradationfn}}
\end{figure}

\begin{rem}\label{rem:degradation}
It can be easily verified that the degradation function $\lambda(\cdot)$ is non-negative, symmetric, i.e., $\lambda(x)=\lambda(-x)$, and periodic with period $2\pi$, i.e, $\lambda(x)=\lambda(x+2\pi k)$ for any integer $k$.
A plot of $\lambda(\cdot)$  is given in Fig.~\ref{fig:degradationfn} for $M\geq 3$.
\end{rem}

Although the general MCRB \eqref{eqn:MCRBfreq} is an exact lower bound, the expression appears to be intractable.
In practice, the matrix $\bm{\Gamma}$ is well approximated by a scaled identity matrix.
For example, based on the simulation conditions in Section~\ref{sec:Simulation-Results-And}, the magnitudes of some rows of $\bm{\Gamma}$ are plotted in Fig.~\ref{fig:near_diag}.
We see that the diagonal elements of $\bm{\Gamma}$ are almost the same, while the off-diagonal elements are very close to zero.
Thus, for our problem of interest, the MCRB in \eqref{eqn:CRB-noCSI} is in fact a good approximation of the exact MCRB given by the non-closed-form \eqref{eqn:MCRBfreq}.
Intuitively, the approximation is accurate because of the following observations. If the matrix $\mathbf{K}$ in \eqref{eqn:noisevec} is a circulant matrix, then from the fact that circulant matrices  are diagonalizable by the Fourier matrix, it can be shown that $\bm{\Gamma}$ equals a scaled identity matrix after performing expectation over $\mathbf{h}_{01}$. Since $\mathbf{K}$ can be written as a circulant matrix plus a sparse matrix with typically small-magnitude entries, we expect that $\mathbf{K}$ is close to being circulant, and both the MCRBs are thus approximately the same.
Further numerical results from Section~\ref{sec:Simulation-Results-And} will support the accuracy of this approximation.

\begin{figure}
\centering{}\includegraphics[scale=\SCALE]{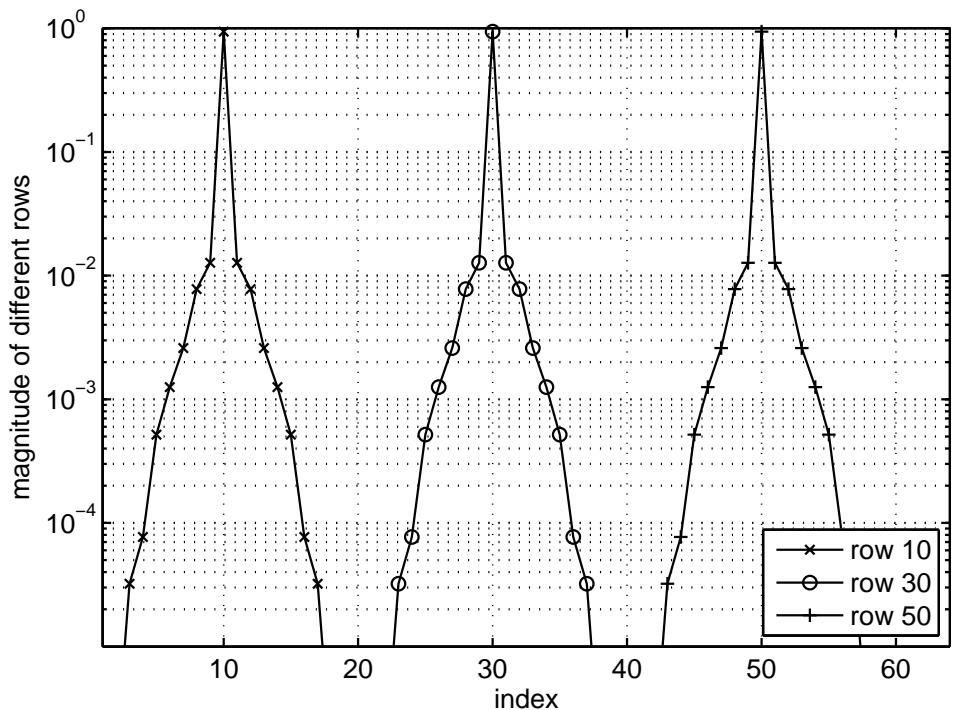}\caption{A plot of the magnitudes of some rows of $\bm{\Gamma}$.  \label{fig:near_diag}}
\end{figure}

Henceforth, we refer to \eqref{eqn:CRB-noCSI} as the approximate CRB (ACRB). As shown in the discussions, the ACRB is in closed-form and provides a good approximation of the MCRB if $\bm{\Gamma}$ is well approximated by a scaled matrix. To obtain further analytical insights, we focus on the ACRB.

Our system model covers the case of one-way relaying viz. Remark~\ref{rem:one-way}. Theorem~\ref{cor:Average Suppose-a-block} states the ACRB for one-way relaying, denoted as $\ACRBoneway$, which serves as the benchmark for two-way relaying.

\begin{theorem}\label{cor:Average Suppose-a-block}
Suppose the BRP is used by source $S_2$ for one-way relaying. Then the ACRB of $\phi_{21}$ is 
\begin{IEEEeqnarray}{rCl}
\ACRBoneway=
\dfrac{6k}{N(N^{2}-1)(\mathbf{r}_{21}^{H}\mathbf{r}_{21})}
\label{eq:average CRB 1-way}
\end{IEEEeqnarray}
and $\ACRBoneway\leq \ACRB$ with equality if and only if the degradation function $\lambda(\cdot)$ equals zero.
\end{theorem}
\begin{proof}
We omit the proof which is similar to that for Theorem~\ref{thm:CRBfinal}. The inequality follows because the degradation function is non-negative.
\end{proof}

The inequality in Theorem~\ref{cor:Average Suppose-a-block} is intuitively expected since in one-way relaying, there is no self-interference to degrade the performance of the estimator.
More interestingly, Theorem~\ref{cor:Average Suppose-a-block} suggests that the ACRB for two-way relaying can achieve the lower bound  if the degradation function can be made to be zero. This observation partly motivates the BRP design in the next section.

\section{Optimization of BRP Parameters}\label{sec:opt}


In this section, we optimize the parameters of the BRP for both sources, namely, the basis blocks $\mathbf{b}_1, \mathbf{b}_2,$ and the block rotation angles $\theta_1, \theta_2$, such that the ACRB is minimized.

We consider the following optimization problem to minimize the ACRB given in \eqref{eqn:CRB-noCSI}:
\begin{IEEEeqnarray}{RLL}\label{eqn:minCRB}
\IEEEyessubnumber\label{eqn:minCRBobj}
\min_{\mathbf{b}_1, \mathbf{b}_2, \theta_1, \theta_2} && \ACRB \\
\IEEEyessubnumber\label{eqn:constraint}
\mbox{subject to } &&  \theta_i\in(0,2\pi] \mbox{ and } \mathbf{b}_i^H \mathbf{b}_i\leq P_i \mbox{ for } i=1,2
\end{IEEEeqnarray}
where $P_i$ represents the power constraint for the BRP sent by source $S_i$.
Although we do not explicitly consider minimizing the peak-to-average power ratio (PAPR) of the BRP, we shall see in Remark~\ref{rem:papr} later the optimal solution also minimizes the PAPR.

From the closed-form solution \eqref{eqn:CRB-noCSI}, the ACRB is proportional to the inverse of $(k-\lambda(\phi_{21}-\phi_{11}))(\mathbf{r}_{21}^{H}\mathbf{r}_{21})$ where $k$ is a positive constant such that the first product is positive.
Thus, the minimization in \eqref{eqn:minCRBobj} depends on these {\em joint} optimization problems 
\begin{IEEEeqnarray}{RCl}\label{eqn:minCRBboth}
\IEEEyessubnumber\label{eqn:minCRB2}
\max_{\mathbf{b}_1, \mathbf{b}_2, \theta_1, \theta_2} &\mathbf{r}_{21}^{H}\mathbf{r}_{21}\\
\IEEEyessubnumber\label{eqn:minCRB1}
\min_{0\leq \Delta < 2\pi} &\lambda(\Delta)
\end{IEEEeqnarray}
where \eqref{eqn:minCRB2} is subject to \eqref{eqn:constraint} and $\Delta\triangleq \theta_{2}-\theta_{1}$ in \eqref{eqn:minCRB1} is the difference of the block rotation angles. 
In \eqref{eqn:minCRB1}, it is sufficient to perform the optimization over $0\leq \Delta< 2\pi$ as $\lambda(\cdot)$ is periodic with period $2\pi$.

Due to the presence of the variables $\theta_{2}, \theta_{1}$ in both optimizations in \eqref{eqn:minCRBboth}, to solve \eqref{eqn:minCRB} optimally we have to consider both optimizations jointly in general.
In Section~\ref{sec:sub:blockbasisopt}, however, we shall see that the optimization in \eqref{eqn:minCRB2} depends only on the basis block $\mathbf{b}_2$; when we optimize the ACRB at source $S_2$ instead of $S_1$ here, the optimization then depends only on $\mathbf{b}_1$ only.
This observation then allows us to decouple the two optimization problems in \eqref{eqn:minCRBboth}.
In Section~\ref{sec:sub:blockangleopt}, we shall thus only consider the optimization of the angle difference $\Delta$ in \eqref{eqn:minCRB1}.
We denote all optimal parameters with the superscript $^{\star}$.

%
%
%

\subsection{Optimization of Basis Blocks}\label{sec:sub:blockbasisopt}

Consider the optimization problem \eqref{eqn:minCRB2} subject to \eqref{eqn:constraint}.
From the definition $r_{21,n} =  h_{21,n}\otimes x_{2n}$ in Section~\ref{sec:System-Model}, 
we can write
$
\mathbf{r}_{21} = \RCSO(\mathbf{b}_2,\theta_2) \mathbf{h}_{12}
$
where
\begin{IEEEeqnarray}{rCl}
\RCSO(\mathbf{b}_2,\theta_2)\triangleq\left[\begin{array}{ccccc}
b_{20} & b_{2,L-1}e^{-j\theta_2} & b_{2,L-2}e^{-j\theta_2} & \cdots & b_{21}e^{-j\theta_2}\\
b_{21} & b_{20} & b_{2,L-1}e^{-j\theta_2} & \cdots & b_{22}e^{-j\theta_2}\\
\vdots & \vdots & \vdots & \ddots & \vdots\\
b_{2,L-1} & b_{2,L-2} & b_{2,L-3} & \cdots & b_{20}
\end{array}\right]
\end{IEEEeqnarray}
depends only on $\mathbf{b}_2,\theta_2$.
Thus the optimization problem \eqref{eqn:minCRB2} becomes
\begin{IEEEeqnarray}{RLL}\label{eqn:minCRB2a}
\IEEEyessubnumber\label{subeqn:minCRB2a1}
\max_{\mathbf{b}_2, \theta_2}  && \left\|\RCSO(\mathbf{b}_2,\theta_2) \mathbf{h}_{12}\right\|^2 \\
\IEEEyessubnumber\label{subeqn:minCRB2a2}
\mbox{subject to } && \theta_2\in(0,2\pi] \mbox{ and } \mathbf{b}_2^H \mathbf{b}_2\leq P_2.
\end{IEEEeqnarray}

Theorem~\ref{thm:optsol1} later states that the optimal solution is given by modifying the well-known CAZAC sequence with some pre-determined phase shifts. A length-$L$ sequence (written as a vector for convenience) $\mathbf{b}=[b_0, \cdots, b_{L-1}]^T$ is said to be a {\em CAZAC sequence} if it satisfies the following properties \cite{Milewski}:
\begin{itemize}
\item constant amplitude: $|b_i|$ equals a constant for $i=0,\cdots,L-1$.
\item cyclic-shift orthogonality: the $i$th cyclic shift of $\mathbf{b}$ is orthogonal to the $j$th cyclic shift of $\mathbf{b}$ for $i\neq j$.
\end{itemize}
Given the CAZAC sequence $\mathbf{b}$ and an angle parameter $\theta$, we define the {\em generalized CAZAC sequence} as
$\mathbf{b}'=[b'_1, \cdots, b'_L]^T$ where $b'_n=b_{n}e^{j n \theta/L}, n=0,\cdots,L-1$.
Clearly, if we choose the angle parameter $\theta=0$, the generalized CAZAC sequence specializes to the conventional CAZAC sequence.


\begin{theorem}\label{thm:optsol1}
For the optimization problem \eqref{eqn:minCRB2a}, the optimal block rotation angle $\theta_2^{\star}$ can take any arbitrary angle, while the optimal basis block $\mathbf{b}_2^{\star}$ is given by the generalized CAZAC sequence with the angle parameter $\theta$ set as the chosen block rotation angle $\theta_2^{\star}$.
\end{theorem}
\begin{proof}
The objective function in \eqref{subeqn:minCRB2a1} can be upper bounded as follows:
\begin{IEEEeqnarray}{RCl}\label{eqn:obj1}
\left\|\RCSO(\mathbf{b}_2,\theta_2) \mathbf{h}_{12}\right\|^2
&=& \left\|\sum_{m=0}^{L-1} \rcso_m h_{12,m}\right\|^2 \\
&\leq& \sum_{m=0}^{L-1} \left\|\rcso_m h_{12,m}\right\|^2 \\
&\leq& \sum_{m=0}^{L-1} \left\| \rcso_m\right\|^2  \sum_{n=0}^{L-1} \left| h_{12,n}\right|^2
= \sum_{m=0}^{L-1} \left| b_{2m}\right|^2  \sum_{n=0}^{L-1} \left| h_{12,n}\right|^2 \\
&\leq&  P  \sum_{n=0}^{L-1} \left| h_{12,n}\right|^2
\end{IEEEeqnarray}
where $\rcso_i$ is the $i$th column of $\RCSO(\mathbf{b},\theta_2)$; the first and second inequalities follow from the triangle inequality and the Cauchy-Schwarz inequality, respectively; the last inequality is due to the constraint \eqref{subeqn:minCRB2a2}.
Now if we use the generalized CAZAC sequence (with power constraint $P$) as the basis block for the BRP, and we set the angle parameter as the block rotation angle $\theta_2$, then we check that the objective function achieves the above bound for any choice of $\theta_2$.
It follows that the stated solutions are optimal.
\end{proof}

\begin{rem}\label{rem:papr}
Since the generalized CAZAC sequence is obtained from the conventional CAZAC sequence with multiplication of a phase term, it retains the desirable property of having constant amplitude. This minimizes the PAPR of the transmitted sequence which is a desirable property in preamble designs, see e.g. \cite{Minn_cazac,stoica_uniform} which consider the preamble design for point-to-point channels.
\end{rem}

\subsection{Optimization of Block-Rotation Angles}\label{sec:sub:blockangleopt}
%

We have shown in Theorem~\ref{thm:optsol1} that the optimization problem \eqref{eqn:minCRB2} does not depend on the the block rotation angles. To solve \eqref{eqn:minCRB} completely, this section solves the optimization problem \eqref{eqn:minCRB1}.

In Theorem~\ref{cor:infCRB}, we observe that the catastrophic case of unbounded CRB occurs if $f_1=f_2$. Hence, we first focus on this case.
Theorem~\ref{thm:optangle_samef} states the necessary and sufficient conditions for $\Delta^{\star}$ to be optimal assuming $f_1=f_2$.
Theorem~\ref{thm:taylor} next considers the case where $f_2$ approaches $f_1$, but not necessarily $f_1=f_2$.

\begin{theorem}
\label{thm:optangle_samef}
If $f_1=f_2$, the necessary and sufficient conditions for $\Delta$, where $0\leq \Delta<2\pi$, to be an optimal solution for the optimization problem \eqref{eqn:minCRB1} are given by
\begin{IEEEeqnarray}{rCl}\label{eq:rotation angles interf freeall}
\Delta&\neq& 0
\IEEEyessubnumber\label{eq:rotation angles interf free1} \\
M\cos(M\Delta/2) \sin(\Delta/2)&=&\sin(M\Delta/2) \cos(\Delta/2).
\IEEEyessubnumber\label{eq:rotation angles interf free}
\end{IEEEeqnarray}
\end{theorem}
\begin{proof}
Denote the numerator and denominator of $\lambda(\Delta)$ in \eqref{eq:definition of lambda} as $N(\Delta)$ and $D(\Delta)$, respectively, where $0\leq \Delta< 2\pi$.
It is useful to observe that $D(\Delta)=0$ iff $\Delta=0$ for $M\geq 3$; this follows from Lemma~\ref{lem:roots} in Appendix~\ref{app:lem:roots} and that $\sin^2(\Delta/2)=0$ iff $\Delta=0$. Note also that $\lambda(\Delta)=N(\Delta)/D(\Delta) =0$ if $N(\Delta)=0$ and $D(\Delta)\neq 0$.

Since $\lambda(\cdot)\geq 0$, to prove Theorem~\ref{thm:optangle_samef}, it suffices to show that $\lambda(\Delta)=0$ iff \eqref{eq:rotation angles interf freeall} holds.
Note that \eqref{eq:rotation angles interf free1} implies $D(\Delta)\neq 0$, while \eqref{eq:rotation angles interf free} implies $N(\Delta)=0$. Thus, \eqref{eq:rotation angles interf freeall} implies $\lambda(\Delta)=0$.

Next, we show the converse, i.e., $\lambda(\Delta)=0$  implies \eqref{eq:rotation angles interf freeall}.
We first show by contradiction that \eqref{eq:rotation angles interf free1}, i.e., $\Delta\neq 0$, must hold, assuming $\lambda(\Delta)=0$.
Suppose that $\Delta=0$. Then $N(\Delta)=0$, and by the well-known L'Hospital's rule $\lambda(0)=\lim_{\Delta\rightarrow 0} N(\Delta)/D(\Delta)$ is strictly positive.
This contradicts the assumption $\lambda(\Delta)=0$.
We conclude that $\Delta\neq 0$ if  $\lambda(\Delta)=0$. Since $\Delta\neq 0$, from the earlier observation $D(\Delta)\neq 0$.
Thus, $\lambda(\Delta)=0$ implies $N(\Delta)=0$ or equivalently \eqref{eq:rotation angles interf free}. Combining the two conditions $\Delta\neq 0$ and \eqref{eq:rotation angles interf free}, we complete the converse part of the proof.
%
%
\end{proof}
\begin{rem}
If $M\geq 3$ is odd, the solution $\Delta^{\star}=\pi$ satisfies \eqref{eq:rotation angles interf freeall}. By Theorem~\ref{thm:optangle_samef}, $\Delta^{\star}$ is an optimal solution of \eqref{eqn:minCRB1} assuming $f_1=f_2$.
Thus, $\Delta^{\star}$ is independent of $M$, which is desirable if $M$ is not known, e.g., due to uncertainty in timing synchronization.
Intuitively, this choice of $\Delta^{\star}$ creates an artificial frequency offset that is maximally possible, since the degradation function is periodic with period $2\pi$.
\end{rem}

\begin{rem}
If $M\geq 4$ is even, no closed-form optimal solution is readily available. For small $M$, we solve \eqref{eq:rotation angles interf freeall} to obtain
\begin{IEEEeqnarray}{rCl}
\Delta^{\star}(M)=\left\{ \begin{array}{ll}
\arccos\left(-2/3\right), & \mbox{for }M=4,\\
\arccos\left(-(9+d_1)/(12+d_1)\right), & \mbox{for }M=6,\\
\arccos\left(-(4+d_2)/(5+d_1)\right), & \mbox{for }M=8,
\end{array}\right.
\end{IEEEeqnarray}
where $d_1=2\sqrt{21}$ and $d_2=2\sqrt{15}\cos\left(\arccos\left(19/\left(5\sqrt{15}\right)\right)/3\right)$.
For arbitrary even $M$, a good approximate solution can be heuristically obtained as $\widetilde{\Delta}(M)=\left(1-M/(M^{2}-1)\right)\pi.$
Numerical results based on $\widetilde{\Delta}(M)$ are  given in Fig.~\ref{fig:approxzero}, which shows that the degradation function $\lambda(\widetilde{\Delta}(M))$ is very close to zero and decreases quickly with increasing $M$.
\end{rem}

\begin{figure}
\centering{}\includegraphics[scale=\SCALE]{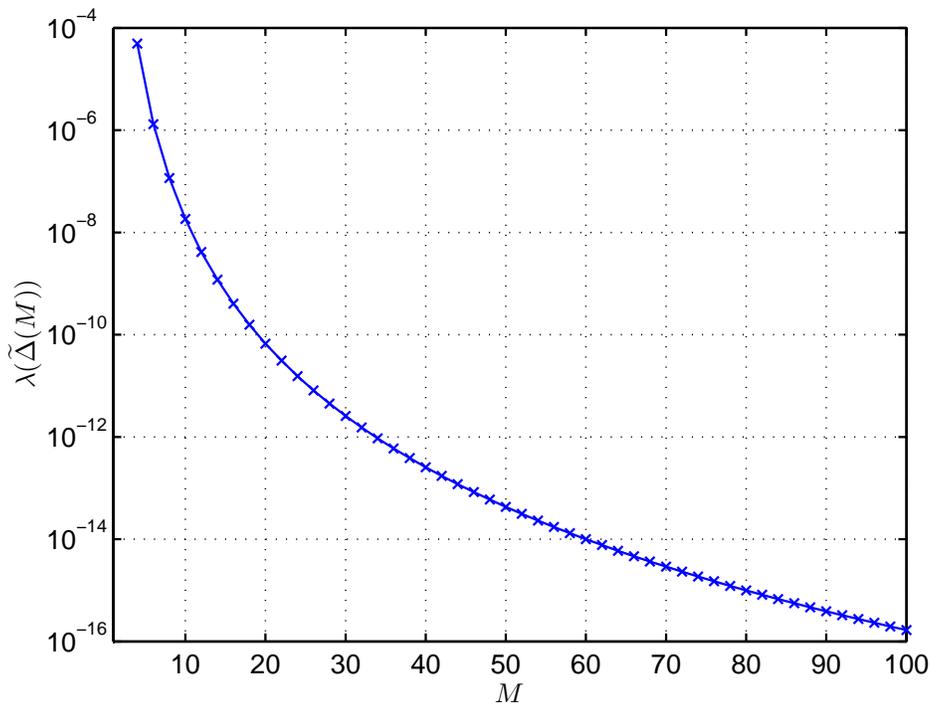}\caption{Graph of degradation function $\lambda(\widetilde{\Delta}(M))$ vs number of basis blocks $M$ for even $M\geq 4$, where $\widetilde{\Delta}(M)=(1-M/(M^{2}-1))\pi$ approximates the optimal solution $\Delta^{\star}(M)$. \label{fig:approxzero}}
\end{figure}

We next consider the case where $f_2\approx f_1$ but $f_2\neq f_1$. We expect the CFO to be small because the carrier frequencies are typically selected to be close to some designated carrier frequency, or a ranging procedure is performed before two-way relaying to calibrate the carrier frequencies.
Theorem~\ref{thm:taylor} obtains the optimal $\Delta^{\star}$ that satisfies \eqref{eqn:minCRB1} such that it also minimizes the (small) perturbation of the degradation function around the neighborhood of $f_1$ and $f_2$.
To obtain an explicit closed-form solution, we focus on $M$ that is odd.
\begin{theorem}\label{thm:taylor}
Let $\delta=2\pi(f_2-f_1)L$ and $\Delta=\theta_2-\theta_1$. 
For small $\delta$, the degradation function is given by 
\begin{IEEEeqnarray}{rCl}\label{eqn:taylorlambda}
\lambda(\Delta+\delta)&=& 
\lambda(\Delta)+ p_1 \delta + p_2 \delta^{2}+\mathcal{O}(\delta^{3}),
\end{IEEEeqnarray}
where $p_1=0$ and $p_2=\dfrac{M^{2}-1}{4\sin^2(\Delta/2)}$.
For odd $M\geq 3$, the optimal $\Delta^{\star}$ that satisfies \eqref{eqn:minCRB1} and minimizes the second-order perturbation term $p_2$ is uniquely given by $\Delta^{\star}=\pi$.
\end{theorem}
\begin{proof}
After some calculus and algebraic manipulations, the first and second derivatives of $\lambda(\Delta)$ can be obtained as $\lambda'(\Delta)=p_1=0$ and $\lambda''(\Delta)=2p_2$, respectively.
The Taylor series of $\lambda(\Delta+\delta)$ at $\Delta$ for small $\delta$ can then be obtained as given in \eqref{eqn:taylorlambda}.
We have $p_2\geq (M^{2}-1)/4$ with equality iff $\sin^2(\Delta/2)=1$, i.e., $\Delta=\pi$. The proof is completed by noting that $\Delta=\pi$ also satisfies the optimality in \eqref{eq:rotation angles interf freeall}.
\end{proof}

From Theorem~\ref{thm:taylor}, by choosing odd $M\geq3$ and $\Delta=\pi$, the degradation function is exactly zero if $f_2=f_1$, and the first-order perturbation term $p_1$ is also zero for small perturbation in the CFO $(f_2-f_1)$. Thus, up to a second order perturbation term of the CFO, the angle difference $\Delta^{\star}=\pi$ is optimal.

%

\section{CRB-Preserving Implementation \label{sec:Effect-of-Filtering}}

In this section, we design an estimator that achieves the ACRB that is optimized in Section~\ref{sec:opt}.
Although the analysis of the CRB for two-way relaying is rather involved, we shall show that CFO estimators that achieve the CRB in the point-to-point channel suffices for the two-way relay channel.

To achieve low complexity in implementation, we introduce a simple preprocessing to reduce the received signal vector to a more familiar form.
Consider the linear filter
$\mathbf{Q}\in\mathbb{C}^{NL\times m}$
with the orthogonality property that
\begin{IEEEeqnarray}{rCl}
\mathbf{Q}^{H}\mathbf{G}_{11}=\mathbf{0}_{m\times L},\label{eq:ZF}
\end{IEEEeqnarray}
where the integer $m$, where $1\leq m\leq (N-1)L$, is a design parameter.
The matrix $\mathbf{Q}$ always exists since the nullity of $\mathbf{G}_{11}^H$ (of rank $L$) is $NL-L$.
Applying the linear filter on the received signal vector $\mathbf{r}_{1}$ in \eqref{eq:discrete model at S1}, we get
\begin{IEEEeqnarray}{rCl}
\mathbf{z}_{1}\triangleq\mathbf{Q}^{H}\mathbf{r}_{1}=\mathbf{Q}^{H}\mathbf{G}_{21}\mathbf{r}_{21}+\mathbf{Q}^{H}\mathbf{u}_{1}.\label{eq:single tone 01}
\end{IEEEeqnarray}

We shall use $\mathbf{z}_{1}$, instead of $\mathbf{r}_{1}$, for CFO estimation.
Despite the low complexity, Theorem \ref{thm:linearZFfilter} states that there is no loss of optimality in the GCRB if $\mathbf{Q}$ a {\em CRB-preserving filter} that satisfies these two conditions: (i) $m=(N-1)L$ and (ii) $\mathbf{Q}$ is of full rank.
\begin{theorem}
\label{thm:linearZFfilter}
Consider the class of linear filter $\mathbf{Q}$ with parameter $m$, where $1\leq m\leq (N-1)L$, such that \eqref{eq:ZF} holds.
If a CRB-preserving filter is used, then the GCRB of $\phi_{21}$ using the received signal $\mathbf{r}_{1}$, given by \eqref{eq:CRB complete}, is the same as the GCRB using instead the filtered signals $\mathbf{z}_{1}$ in \eqref{eq:single tone 01}.
\end{theorem}
\begin{proof}
Denote the function
$
f(\mathbf{X})\triangleq\left(2\mathbf{r}_{21}^{H}\mathbf{G}_{21}^{H}\mathbf{T}\left[\mathbf{X}-\mathbf{X}\mathbf{G}_{21}\left(\mathbf{G}_{21}^{H}\mathbf{X}\mathbf{G}_{21}\right)^{-1}\mathbf{G}_{21}^{H}\mathbf{X}\right]\mathbf{T}\mathbf{G}_{21}\mathbf{r}_{21}\right)^{-1}.
$
Following the proof in Appendix \ref{app:Proof-of-Theorem 1},
the GCRB of $\phi_{21}$ using $\mathbf{z}_{1}$
in (\ref{eq:single tone 01}) can be derived as $f(\mathbf{\overline{\mathbf{\Phi}}}_{1})$,
where $\overline{\mathbf{\Phi}}_{1}=\mathbf{Q}\left(\mathbf{Q}^{H}\mathbf{R}\mathbf{Q}\right)^{-1}\mathbf{Q}^{H}.$
The GCRB in (\ref{eq:CRB complete}) can be alternatively expressed as $\GCRB=f(\mathbf{\Omega}_{1})$,
where $\mathbf{\Omega}_{1}\triangleq\mathbf{R}^{-1}-\mathbf{R}^{-1}\mathbf{G}_{11}\left(\mathbf{G}_{11}^{H}\mathbf{R}^{-1}\mathbf{G}_{11}\right)^{-1}\mathbf{G}_{11}^{H}\mathbf{R}^{-1}$.
Consider the difference ${\bm{\Delta}}\triangleq\mathbf{\Omega}_{1}-\overline{\mathbf{\Phi}}_{1}=\mathbf{R}^{-1}-\mathbf{X}_{1}\left(\mathbf{X}_{1}^{H}\mathbf{R}\mathbf{X}_{1}\right)^{-1}\mathbf{X}_{1}^{H}$,
where $\mathbf{X}_{1}\triangleq\left[
\mathbf{R}^{-1}\mathbf{G}_{11},  \mathbf{Q}\right]\in\mathbb{C}^{NL\times(L+m)}$. Since $\mathbf{R}^{-1}$ is of rank $NL$,
whereas $\mathbf{X}_{1}\left(\mathbf{X}_{1}^{H}\mathbf{R}\mathbf{X}_{1}\right)^{-1}\mathbf{X}_{1}^{H}$
is of rank $L+m$, the sufficient and necessary condition for ${\bm{\Delta}}$
to be zero is to make $\mathbf{X}_{1}$ full rank, or equivalently,
$m=(N-1)L$, and $\mathbf{Q}$ is of full rank. When these two conditions hold, $\mathbf{X}_{1}$ is invertible and ${\bm{\Delta}}=\mathbf{0}_{NL\times NL}$. Thus $\overline{\mathbf{\Phi}}_{1}=\mathbf{\Omega}_{1}$, implying the two GCRBs are the same.
\end{proof}


Theorem~\ref{thm:linearZFfilter} suggests that the filtered signal $\mathbf{z}_1$ is as good as the original received signal $\mathbf{r}_1$ for CFO estimation, assuming knowledge of the channel $\mathbf{h}_{01}$. This is supported by numerical results in Section~\ref{sec:Simulation-Results-And} even if $\mathbf{h}_{01}$ is not known.

Next, we give a specific realization of the CRB-preserving filter. Any $\mathbf{Q}$ that spans the nullspace of $\mathbf{G}_{11}^H$ must satisfy \eqref{eq:ZF}.
Thus there are infinitely many possible matrices $\mathbf{Q}$ that result in no loss in CRB.
A specific choice of filter for source $S_1$ that satisfies the conditions in Theorem \ref{thm:linearZFfilter} is
\begin{IEEEeqnarray}{rCl}
\mathbf{\widetilde{Q}}_{1}^{H}=\left[\begin{array}{ccccc}
\rho_{11}\mathbf{I}_{L} & -\mathbf{I}_{L} & \mathbf{0}_L & \cdots & \mathbf{0}_L\\
\mathbf{0}_L & \rho_{11}\mathbf{I}_{L} & -\mathbf{I}_{L} & \ddots& \vdots\\
 \vdots & \ddots & \ddots & \ddots & \mathbf{0}_L\\
\mathbf{0}_L & \cdots & \mathbf{0}_L & \rho_{11}\mathbf{I}_{L} & -\mathbf{I}_{L}.
\end{array}\right]\label{eq:Filter Q}
\end{IEEEeqnarray}
We can interpret the filter $\mathbf{\widetilde{Q}}$ as a blockwise low pass filter with coefficients $\{\rho_{11},-1\}$ that operate on the received signal in two blocks
of $L$ samples.
%
The advantage of using this filter is that it leads to low implementation complexity, comparable to point-to-point communication systems.
Observe that the filtered vector $\mathbf{z}_{1}$ can be interpreted to have an equivalent channel $\mathbf{\widetilde{Q}}^{H}\mathbf{G}_{21}=(1-\rho_{21}) \mathbf{G}_{21}$ with an equivalent Gaussian noise $\mathbf{\widetilde{Q}}_{1}^{H}\mathbf{u}_{1}$.
The signal model is the same as a point-to-point communication system where a periodic preamble is sent, experiences a CFO of $\theta_2$, and is received with Gaussian noise with correlation matrix $\mathbf{\widetilde{Q}}_{1}^{H}\mathbf{\widetilde{Q}}_{1}$. Hence, we can use any CRB-achieving estimator for CFO estimation in point-to-point communication, and yet achieves the CRB for two-way relaying.

\section{Simulation Results And Discussions}\label{sec:Simulation-Results-And}

Without loss of generality, we consider the estimation at source $S_1$.
In this section, we shall see that the proposed BRP design can achieve a mean-squared-error (MSE) performance that is close to the EMCB, which is the fundamental lower bound according to Remark~\ref{rem:boundedcond}. Specifically, the EMCB is obtained numerically by averaging the GCRB assuming knowledge of $\mathbf{h}_{01}$ is available.
To obtain the numerical results for the MSE, we consider two specific estimators, namely the {\em correlator}, see e.g. \cite{popular_correlator}, and the {\em genie-aided maximum likelihood estimator} (GA-MLE). They represent schemes with very low and very high complexity, respectively, and both are commonly used. Both estimators work on the output
$\mathbf{z}_{1}$ of the blockwise linear filter in (\ref{eq:Filter Q}),
obtained based on (\ref{eq:single tone 01}), to take advantage of the fact that the signal is free of self-interference and the link becomes a point-to-point channel.

The correlator estimates $\phi_{21}$ as $\hat{\phi}_{\mathsf{COR}}=\angle\sum_{m=0}^{M-3}\mathbf{z}_{1,m}^{H}\mathbf{z}_{1,m+1}$,
where $\mathbf{z}_{1,m}$ is a column vector that collects elements
of $\mathbf{z}_{1}$ in (\ref{eq:single tone 01}) with indices $\{mL,\ldots,(m+1)L-1\}$.
The GA-MLE, on the other hand, is given the channel $\mathbf{h}_{01}$ and hence knows the noise covariance matrix given by \eqref{eqn:noisecov}. Thus the noise $\mathbf{u}_1$ is known to be Gaussian distributed, and this estimator then performs conventional ML estimation to obtain the CFO.
%
In practice, $\mathbf{h}_{01}$ is not known and hence the estimator provides an optimistic MSE performance that may not be realized in practice.

We assume the following simulation setup.
We use $M=5$ basis blocks and $L=16$ samples in each basis block.
All channel taps experience independent Rayleigh
fading of magnitude regulated by a $8$-tap exponential power delay profile
$k e^{-n/\tau_{\textrm{rms}}}$, where $n$ is the tap index and $k$ is a normalizing constant.
We choose $\tau_{\textrm{rms}}=1$.
For simplicity, it is assumed
that the two sources $S_{1}$ and $S_{2}$ communicate with equal power $P$,
and the variance of the AWGN $\sigma^2_n$ at all receivers are identical.
The signal-to-noise ratio is defined as $P/\sigma^2_n$. The scaling factor at the relay is
set such that the total transmission power is $2P$.
The carrier frequencies of the two sources relative to the carrier frequency of the relay are (arbitrarily) set to $f_{1}=0.001$ and $f_{2}=-0.002$. 

In Figure~\ref{fig:MSE for NCFO=00003D5,rms=00003D1}, we plot the EMCB given by averaging the GCRB for the following three cases:
(i) two-way relaying with periodic preambles (with no marker);
(ii) two-way relaying with  BRP optimized according to Section~\ref{sec:opt} (with marker $``+"$);
(iii) one-way relaying with periodic preambles (with marker $``\times"$).
We observe that the first case has fairly large averaged GCRB, which is expected according to Theorem~\ref{cor:infCRB}.
With the optimized BRP, however, the averaged GCRB is reduced substantially by about twenty times.
Moreover, the averaged GCRB for both two-way relaying and one-way relaying are almost the same when the optimized BRP is used, as suggested by Theorem~\ref{cor:Average Suppose-a-block} based on the ACRB.

Theorem~\ref{thm:linearZFfilter} states that a CRB-preserving filter with the optimized BRP allow the GCRB to be approached. To check this, in Figure~ \ref{fig:MSE for NCFO=00003D5,rms=00003D1} we also plot the MSE for two-way relaying where we use the optimized BRP and the following estimators:
(i) correlator (with marker $``\square"$);
(ii) GA-MLE (with marker $``\circ"$).
We observe that the GA-MLE performs close to the averaged GCRB for two-way relaying with optimized BRP. This is expected when knowledge of $\mathbf{h}_{01}$ is available, according to Theorem~\ref{thm:linearZFfilter}. At high SNR, the correlator performs close to the GA-MLE. This suggests that the additional knowledge of the channel $\mathbf{h}_{01}$ does not improve the MSE performance, and so a low complexity estimator suffices. We note that the correlator outperforms the averaged GCRB at SNRs lower than about $7$~dB. This is because the correlator is not an unbiased estimator as we have numerically confirmed. Nevertheless, in the high-SNR regime the correlator becomes asymptotically unbiased, and so the MSE still becomes lower bounded by the averaged GCRB.

\begin{figure}
\centering{}\includegraphics[scale=\SCALE]
{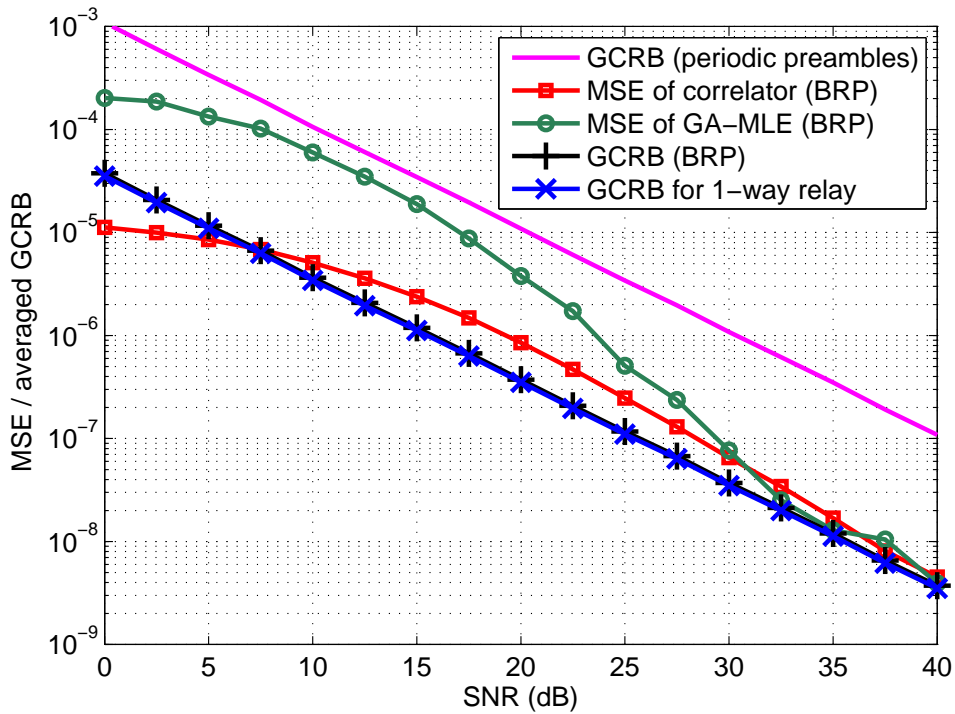}
\caption{Mean-square-error (MSE) and the averaged GCRB for CFO estimation.
\label{fig:MSE for NCFO=00003D5,rms=00003D1}}
\end{figure}

\section{Conclusions\label{sec:Conclusions}}

A novel block-rotation-based preamble (BRP) design for CFO estimation in amplify-and-forward
two-way relaying systems has been proposed. The BRP can be viewed as a generalization of the conventional periodic preamble widely used in practice.
Intuitively, the BRP creates an artificial block-level frequency offset so as to distinguish the carrier frequencies of the two sources. Our analysis on the fundamental lower bound of the MSE performance allows us to identify the catastrophic case when the CRB is unbounded or fails to exist, which has not been identified in the literature so far. Also, our analysis provides practical guidelines to design BRPs that perform close to the fundamental lower bound.
Finally, since the carrier frequency of the relay does not affect the analysis, the proposed BRP design and estimation schemes appear to be readily applicable to communication systems with more than one relay, and also to multiple-antenna communication systems.

\appendices{}
\section{Proof of Theorem~\ref{thm:The-CRB-for}} \label{app:Proof-of-Theorem 1}

For convenience, let
$\mathbf{W}  \triangleq
\left[\begin{array}{cr}
\hspace{-0.1cm} \Re(\mathbf{R}^{-1}) &\hspace{-0.2cm} -\Im(\mathbf{R}^{-1})\hspace{-0.1cm}\\
\hspace{-0.1cm} \Im(\mathbf{R}^{-1}) &\hspace{-0.2cm} \Re(\mathbf{R}^{-1}) \hspace{-0.1cm}
\end{array}\right]$,
$\mathbf{q}  \triangleq  \frac{\partial\mathbf{A}}{\partial\phi_{21}}\cdot \nuiv=\left[\begin{array}{c}
\Re(j\mathbf{T}\mathbf{G}_{21}\mathbf{r}_{21})\\
\Im(j\mathbf{T}\mathbf{G}_{21}\mathbf{r}_{21})
\end{array}\right]
$
and $\mathbf{T} \triangleq \mbox{diag}(0,1,\cdots, M-1)\otimes\mathbf{I}_{L}$.

The parameters to be estimated are $\bm{\param}=[\phi_{21},\mathbf{r}_{11}^T,\mathbf{r}_{21}^T]^{T}$.
Given $\bm{\param}$, and since $\mathbf{h}_{01}$ is given, the received signal $\mathbf{r}_1$ in \eqref{eq:discrete model at S1_real} has the Gaussian distribution
\begin{IEEEeqnarray}{rCl}
f_{\mathbf{r}_{1}}(\mathbf{r}_{1}; \bm{\param})=\pi^{-N}\det(\mathbf{W})\exp\left[-\left(\mathbf{r}-\mathbf{A}\nuiv\right)^{H}\mathbf{W}\left(\mathbf{r}-\mathbf{A}\nuiv\right)\right].
\end{IEEEeqnarray}
The Fisher Information Matrix (FIM) is thus
\begin{IEEEeqnarray}{rCl}\label{eqn:FIM}
\FIM(\mathbf{h}_{01})\triangleq
\mathbb{E}_{\mathbf{r}_{1}}\left(\left[\begin{array}{c}
\dfrac{\partial\log\left[f_{\mathbf{r}_{1}}(\mathbf{r}_{1};\bm{\param})\right]}{\partial\phi_{21}}\\
\dfrac{\partial\log\left[f_{\mathbf{r}_{1}}(\mathbf{r}_{1};\bm{\param})\right]}{\partial\nuiv}
\end{array}\right]\left[\begin{array}{c}
\dfrac{\partial\log\left[f_{\mathbf{r}_{1}}(\mathbf{r}_{1};\bm{\param})\right]}{\partial\phi_{21}}\\
\dfrac{\partial\log\left[f_{\mathbf{r}_{1}}(\mathbf{r}_{1};\bm{\param})\right]}{\partial\nuiv}
\end{array}\right]^{H}\right)=2\left[\begin{array}{cr}
\mathbf{q}^{H}\mathbf{W}\mathbf{q} & \mathbf{q}^{H}\mathbf{W}\mathbf{A} \\
\mathbf{A}^{H}\mathbf{W}\mathbf{q} & \mathbf{A}^{H}\mathbf{W}\mathbf{A}
\end{array}\right].
\IEEEeqnarraynumspace
\end{IEEEeqnarray}
The CRB matrix, given by the inverse of the FIM, can then be obtained. The $(1,1)$th element of the CRB matrix can be isolated,
by the use of the block matrix inversion lemma, to give
\begin{IEEEeqnarray}{rCl}\label{app:eqn:GCRB}
\GCRB =\frac{1}{2\mathbf{q}^{H}\left[\mathbf{W}-\mathbf{W}\mathbf{A}\left(\mathbf{A}^{H}\mathbf{W}\mathbf{A}\right)^{-1}\mathbf{A}^{H}\mathbf{W}\right]\mathbf{q}}.\label{eq:CRB version 0}
\end{IEEEeqnarray}
After some algebraic manipulations, we have
\begin{IEEEeqnarray}{rCl}\label{eq:a}
\mathbf{W}\mathbf{A}\left(\mathbf{A}^{H}\mathbf{W}\mathbf{A}\right)^{-1}\mathbf{A}^{H}\mathbf{W}=\left[\begin{array}{cr}
\Re(\mathbf{M}) & -\Im(\mathbf{M})\\
\Im(\mathbf{M}) & \Re(\mathbf{M})
\end{array}\right]
\end{IEEEeqnarray}
where
\begin{IEEEeqnarray}{rCl}\nonumber
\mathbf{M} & \triangleq & \mathbf{R}^{-1}\left[\begin{array}{cc}
\mathbf{G}_{11} & \mathbf{G}_{21}\end{array}\right]\left(\left[\begin{array}{c}
\mathbf{G}_{11}^{H}\\
\mathbf{G}_{21}^{H}
\end{array}\right]\mathbf{R}^{-1}\left[\begin{array}{cc}
\mathbf{G}_{11} & \mathbf{G}_{21}\end{array}\right]\right)^{-1}\left[\begin{array}{c}
\mathbf{G}_{11}^{H}\\
\mathbf{G}_{21}^{H}
\end{array}\right]\mathbf{R}^{-1}\\ \label{eq:b}
 & = & \mathbf{R}^{-1}\mathbf{G}_{21}\left(\mathbf{G}_{21}^{H}\mathbf{R}^{-1}\mathbf{G}_{21}\right)^{-1}\mathbf{G}_{21}^{H}\mathbf{R}^{-1}+\mathbf{\Phi}_{1}\mathbf{G}_{11}\left(\mathbf{G}_{11}^{H}\mathbf{\Phi}_{1}\mathbf{G}_{21}\right)^{-1}\mathbf{G}_{11}^{H}\mathbf{\Phi}_{1}
\label{eq:c}
\end{IEEEeqnarray}
and $\mathbf{\Phi}_1$ is given in \eqref{eq:one-way relay condition0}.
From \eqref{eq:one-way relay condition0} and \eqref{eq:CRB version 0}--\eqref{eq:c}, we obtain (\ref{eq:CRB complete}) and
(\ref{eq:one-way relay condition}).
From \eqref{eq:one-way relay condition}, clearly $\mathbf{\Phi}_{1}$ satisfies $\mathbf{\Phi}_{1} \mathbf{G}_{21}=\mathbf{0}_{NL\times L}$.

\section{Proof of Theorem~\ref{cor:infCRB}}\label{app:infCRB}
From the Taylor's series $e^{j\theta}=1+j\theta+ \mathcal{O}(\theta^2)$ for small $\theta$, we can express
\begin{IEEEeqnarray}{rCl}\label{eqn:111}
\mathbf{G}_{11}=\left(\mathbf{I}-j(\phi_{21}-\phi_{11}) \mathbf{T}\right)\mathbf{G}_{21} + \mathcal{O}\left((\phi_{21}-\phi_{11})^2 \mathbf{I}_{NL\times L}\right).
\end{IEEEeqnarray}
From Theorem~\ref{thm:The-CRB-for} we have $\mathbf{\Phi}_{1} \mathbf{G}_{21}=\mathbf{0}_{NL\times L}$.  Thus from \eqref{eqn:111} we get
\begin{IEEEeqnarray}{rCl}
{\mathbf{G}}_{11}^{H}\mathbf{\Phi_{1}}{\mathbf{G}}_{11} & = & (\phi_{21}-\phi_{11})^{2}\mathbf{G}_{21}^{H}\mathbf{T}\mathbf{\Phi}_{1}\mathbf{T}\mathbf{G}_{21}
+\mathcal{O}\left((\phi_{21}-\phi_{11})^3 \mathbf{I}_L\right)\\
{\mathbf{G}}_{11}^{H}\mathbf{\Phi_{1}}\mathbf{T}\mathbf{G}_{21} & = & j(\phi_{21}-\phi_{11})\mathbf{G}_{21}^{H}\mathbf{T}\mathbf{\Phi}_{1}\mathbf{T}\mathbf{G}_{21}
+\mathcal{O}\left((\phi_{21}-\phi_{11})^2 \mathbf{I}_L\right),
\end{IEEEeqnarray}
which imply $\mathbf{\Psi}_1=\mathbf{G}_{21}^{H}\mathbf{T}\mathbf{\Phi}_{1}\mathbf{T}\mathbf{G}_{21}+\mathcal{O}\left(|\phi_{21}-\phi_{11}| \mathbf{I}_L\right)$.
Since $\theta_{1}=\theta_{2}=0$, the denominator of $\GCRB$
in (\ref{eq:CRB complete}), i.e., $2\mathbf{r}_{21}^{H}\left(\mathbf{G}_{21}^{H}\mathbf{T}\mathbf{\Phi}_{1}\mathbf{T}\mathbf{G}_{21}-{\mathbf{\Psi}}_{1}\right)\mathbf{r}_{21}$,  approaches zero as $f_{2}- f_{1}=(\phi_{21}-\phi_{12})/(2\pi L)$ approaches zero.
Thus, we have proved $\GCRB\rightarrow\infty$ as $(f_{2}- f_{1})\rightarrow 0$, which holds for any $\nuiv$.

\section{Proof of Corollary~\ref{cor:infCRB2}}\label{app:infCRB2}

Suppose $M=2$, thus $\mathbf{G}_{21}=[\rho_{21} \mathbf{I}_L, \rho_{21}^2 \mathbf{I}_L]^T.$
Let $\mathbf{P}=\frac{1}{2} \left[\mathbf{G}_{21} \mathbf{\breve{G}}_{21}\right]$, where $\breve{\mathbf{G}}_{21}=[ \mathbf{I}_L, -\rho_{21} \mathbf{I}_L]^T$.
Clearly $\mathbf{P}\mathbf{P}^H=\mathbf{I}_{2L}$.
Also, let $\mathbf{\Phi}_{1}=\mathbf{Q}\mathbf{D}\mathbf{Q}^H$ be the SVD of $\mathbf{\Phi}_{1}$.
We first show that we can express
\begin{IEEEeqnarray}{rCl}\label{eqn:223}
\mathbf{\Phi}_1=\breve{\mathbf{G}}_{21} \mathbf{B}_{21} \breve{\mathbf{G}}_{21}^H
\end{IEEEeqnarray}
where  $\mathbf{B}_{21}\triangleq \frac{1}{4} \left(\mathbf{W}_{21}\mathbf{D}_{1}\mathbf{W}_{21}^H
+ \mathbf{W}_{22}\mathbf{D}_{2}\mathbf{W}_{22}^H\right)$.

Since both $\mathbf{P}, \mathbf{Q}$ are unitary matrices, there always exists a unitary matrix $\mathbf{W}$ such that $\mathbf{Q}=\mathbf{P}\mathbf{W}$.
It can then be shown that $\mathbf{Q}^H \mathbf{\Phi}_1\mathbf{G}_{21} = \mathbf{D} \mathbf{Q}^H \mathbf{G}_{21}
= \mathbf{D} \mathbf{W}^H \mathbf{P}^H \mathbf{G}_{21}$.
From Theorem~\ref{thm:The-CRB-for}, we know that $\mathbf{\Phi}_1\mathbf{G}_{21}=\mathbf{0}_{2L\times L}$, i.e., $\mathbf{D} \mathbf{W}^H \mathbf{P}^H \mathbf{G}_{21}=\mathbf{0}_{2L\times L}$. It can be easily verified that $\mathbf{P}^H\mathbf{G}_{21} =\left[\mathbf{I}_L,  \mathbf{0}_L\right]^T$, thus we get
\begin{IEEEeqnarray}{rCl}\label{eqn:222}
\mathbf{D}_i \mathbf{W}_{1i}^H=\mathbf{0}_L, i=1,2,
\end{IEEEeqnarray}
where
$\mathbf{D}=
\left[\begin{array}{cr}
\mathbf{D}_1 & \mathbf{0}_L\\
\mathbf{0}_L & \mathbf{D}_2
\end{array}\right] \in \mathbb{C}^{2L\times 2L}$
and
$\mathbf{W}=
\left[\begin{array}{cr}
\mathbf{W}_{11} & \mathbf{W}_{12}\\
\mathbf{W}_{21} & \mathbf{W}_{22}
\end{array}\right] \in \mathbb{C}^{2L\times 2L}.$
Substituting $\mathbf{Q}=\mathbf{P}\mathbf{W}$ and \eqref{eqn:222} into $\mathbf{\Phi}_{1}=\mathbf{Q}\mathbf{D}\mathbf{Q}^H$ then leads to \eqref{eqn:223}.

Using \eqref{eqn:223}, it can then be verified that $\mathbf{G}_{11}^{H}\mathbf{\Phi}_{1}\mathbf{G}_{11} = \left|1-e^{j(\phi_{21}-\phi_{11})}\right|^{2}\mathbf{B}_{21}$.
It follows that $\mathbf{\Psi}_1=\mathbf{G}_{21}^{H}\mathbf{T}\mathbf{\Phi}_{1}\mathbf{T}\mathbf{G}_{21}$, i.e., the denominator of $\GCRB$
in (\ref{eq:CRB complete}) equals zero, independent of $\mathbf{h}_{01}$.
Thus the CRB is undefined for any channel $\mathbf{h}_{01}$.

\section{Proof of \eqref{eqn:FIM2} in Lemma~\ref{lem:1}}\label{app:FIM2}

Let $\widetilde{\bm{\epsilon}}= \mathbf{H} \bm{\epsilon}.$
We write $g(\vparamreal)= \log f_{\mathbf{r}| \mathbf{H}}(\mathbf{r};\vparamreal|\mathbf{H})$ as
\begin{IEEEeqnarray}{rCl} \label{eqn:MFIM0}
g(\vparamreal)
&=& \log f_{\widetilde{\bm{\epsilon}}}(\mathbf{r}-\mathbf{z}(\vparamreal)| \mathbf{H}) \\
&=& \log |J(\mathbf{H})| \times f_{{\bm{\epsilon}}}(\mathbf{H}^{-1} (\mathbf{r}-\mathbf{z}(\vparamreal))^{-1}) \label{eqn:MFIM1} \\
&=& \log |J(\mathbf{H})| + \sum_{n=1}^{N} \log f_{\epsilon_n}(\epsilon_n) \label{eqn:MFIM2}
\end{IEEEeqnarray}
where \eqref{eqn:MFIM0} follows from the independence of the random variables, \eqref{eqn:MFIM1} follows from the transformation of $\bm{\epsilon}$ to $\widetilde{\bm{\epsilon}}$ via the full-rank $\mathbf{H}$ and $J(\mathbf{H})$ is the corresponding Jacobian, and \eqref{eqn:MFIM2} follows from the i.i.d. assumption of $\bm{\epsilon} =\mathbf{r}-\mathbf{z}(\vparamreal)^{-1}$.
After some algebraic manipulations, we get
\begin{IEEEeqnarray}{rCl}\label{eqn:MFIM3}
\frac{\partial g(\vparamreal)}{\partial \vparamreal} = -\sum_{n=1}^{N} \frac{f'_{\epsilon_n}(\epsilon_n)}{f_{\epsilon_n}(\epsilon_n)} \mathbf{Z}'(\vparamreal)\widetilde{\mathbf{h}}_n
\end{IEEEeqnarray}
where we denote $f'_{\epsilon_n}(\epsilon_n)= \partial f_{\epsilon_n}(\epsilon_n)/\partial \epsilon_n $ and $\widetilde{\mathbf{h}}_n^T$ is the $n$th row of $\mathbf{H}^{-1}$, while the definition of $\mathbf{Z}'(\vparamreal)$ appears after \eqref{eqn:FIM2}.

We now obtain the modified FIM $\MFIM=\mathbb{E}_{\mathbf{r}}
\left[\dfrac{\partial g(\bm{\param})}{\partial\vparamreal}
\dfrac{\partial g(\vparamreal)}{\partial\vparamreal^T}\right]$.
Let $\bm{\rho}_n \triangleq  \frac{f'_{\epsilon_n}(\epsilon_n)}{f_{\epsilon_n}(\epsilon_n)} \widetilde{\mathbf{h}}_n$.
Note that $\mathbb{E}_{{\mathbf{H},\bm{\epsilon}}}[\bm{\rho}_n]=
\mathbb{E}_{\epsilon}\left[\frac{f'_{\epsilon_n}(\epsilon_n)}{f_{\epsilon_n}(\epsilon_n)} \right]
\mathbb{E}_{\mathbf{H}}\left[\widetilde{\mathbf{h}}_n\right]
=
\mathbf{0}_{N\times 1}$ due to $\mathbb{E}_{\epsilon}\left[\frac{f'_{\epsilon_n}(\epsilon_n)}{f_{\epsilon_n}(\epsilon_n)} \right]
= \int_{-\infty}^{\infty} f'_{\epsilon_n}(\epsilon_n)  \,\mathrm{d} \epsilon_n = \partial (\int_{-\infty}^{\infty} f_{\epsilon_n}(\epsilon_n) \,\mathrm{d} \epsilon_n)/\partial \epsilon_n =0$.
Using this observation and \eqref{eqn:MFIM3}, the modified FIM simplifies as
\begin{IEEEeqnarray}{rCl}\label{eqn:FIM3}
\MFIM(\vparamreal) = \mathbf{Z}'(\vparamreal) \cdot \mathbb{E}_{{\mathbf{H},\bm{\epsilon}}}\left[\sum_{n=1}^{N} \bm{\rho}_n \bm{\rho}_n^T \right ]\cdot {\mathbf{Z}'}^T (\vparamreal).
\end{IEEEeqnarray}
In \eqref{eqn:FIM3}, we use the fact that $\mathbf{Z}'(\vparamreal)$ is independent of $\mathbf{H}$ and $\bm{\epsilon}$.
The middle matrix can be written as
\begin{IEEEeqnarray}{rCl}\label{eqn:}
\mathbb{E}_{{\mathbf{H},\bm{\epsilon}}}\left[\sum_{n=1}^{N} \bm{\rho}_n \bm{\rho}_n^T \right ]
&=& \sum_{n=1}^{N} \mathbb{E}_{{\mathbf{H},\bm{\epsilon}}}\left[\left(\frac{f'_{\epsilon_n}(\epsilon_n)}{f_{\epsilon_n}(\epsilon_n)}\right)^2 \widetilde{\mathbf{h}}_n \widetilde{\mathbf{h}}^T_n \right]
= \gamma \bm{\Gamma}
\end{IEEEeqnarray}
where $\gamma \triangleq \mathbb{E}_{\epsilon}\left[\left(\frac{f'_{\epsilon_n}(\epsilon_n)}{f_{\epsilon_n}(\epsilon_n)}\right)^2 \right] $ and
$\bm{\Gamma}\triangleq \sum_{n=1}^{N} \mathbb{E}_{{\mathbf{H}}}\left[ \widetilde{\mathbf{h}}_n \widetilde{\mathbf{h}}^T_n \right]$, and it can be readily checked that $\bm{\Gamma}=\mathbb{E}_{\mathbf{H}}\left[\left(\mathbf{H} \mathbf{H}^H\right)^{-1}\right]$. Hence, we obtain \eqref{eqn:FIM2b}.

\section{Proof of Theorem~\ref{thm:CRBfinal}}\label{app:CRBfinal}


We use the substitution $\widetilde{\bm{\param}}=[\phi_{21},\nuiv^T]^T$ and
$\mathbf{z}(\bm{\paramreal})=\mathbf{A}\nuiv$ as in \eqref{eq:discrete model at S1_real}.
Next, we express the noise vector $\mathbf{u}_1$ in \eqref{eq:discrete model at S1} as $\mathbf{H} \bm{\epsilon}$, such that assumptions $A1$-$A3$ hold.
Given $\mathbf{h}_{01}$, the noise vector in \eqref{eqn:noisevec} is Gaussian distributed with zero mean and covariance matrix $\mathbf{R}(\mathbf{h}_{01})=  \mathbf{K}(\mathbf{h}_{01})  \mathbf{K}^H(\mathbf{h}_{01})+\mathbf{I}$ which is full rank with probability one. Let the eigendecomposition of the covariance matrix be $\mathbf{R}(\mathbf{h}_{01})=\mathbf{U}\bm{\Lambda}\mathbf{U}^H$; also let $\bm{\Lambda}^{1/2}$ be a diagonal matrix with diagonal elements given by the square root of the corresponding diagonal elements in $\bm{\Lambda}$.
Without loss of generality, we can express $\mathbf{u}_1$ in \eqref{eqn:noisevec} instead as
\begin{IEEEeqnarray}{rCl}
\mathbf{u}_1=\mathbf{H}(\mathbf{h}_{01}) \bm{\epsilon},
\end{IEEEeqnarray}
where $\mathbf{H}(\mathbf{h}_{01})
=\mathbf{U}\bm{\Lambda}^{1/2}$ and $\bm{\epsilon}$ is an i.i.d. zero-mean unit-variance complex-valued Gaussian vector that is independent of $\mathbf{h}_{01}$.
This is because both representations of $\mathbf{u}_1$ are statistically equivalent.
Taking $\mathbf{h}_{01}$ to be random in general, we see that assumptions $A1$ to $A4$ always hold.
Applying the system model in \eqref{eqn:signalmodel_corrnoise}, the MFIM is given by \eqref{eqn:FIM2b} where after some tedious but straightforward algebraic manipulations, we obtain $\bm{\Gamma}=\mathbb{E}_{\mathbf{h}_{01}}[\mathbf{R}^{-1}(\mathbf{h}_{01})]$.
This proves the first part of Theorem~\ref{thm:CRBfinal}.

For the second part of the proof, suppose that $\bm{\Gamma}=k\mathbf{I}$ for some constant $k>0$. 
Then the MFIM and the MCRB can be obtained similarly as given by the FIM \eqref{eqn:FIM} and GCRB \eqref{app:eqn:GCRB} in Appendix~\ref{app:Proof-of-Theorem 1}, respectively,  with the covariance matrix $\mathbf{R}$ replaced by $(1/k)\mathbf{I}$.
With this substitution $\mathbf{R}=(1/k)\mathbf{I}$, after some tedious but straightforward algebraic manipulations, we obtain the closed-form expression \eqref{eqn:CRB-noCSI}.

\section{An Auxillary Lemma to Prove Theorem~\ref{thm:optangle_samef}}\label{app:lem:roots}
\begin{lem}\label{lem:roots}
For integer $M\geq 3$, $M^2 \sin^2(x/2) = \sin^2(Mx/2)$ holds iff $x\in \mathcal{S}\triangleq \{x=2k\pi, k\in \mathbb{Z}\}$.
\end{lem}
\begin{proof}
We consider the case of odd $M$ and even $M$ separately.

Assume that integer $M\geq 3$ is odd.
We use the well-known identity of the Dirichlet kernel
$
\frac{\sin\left((n+1/2)x\right)}{\sin(x/2)}=1+2\sum_{k=1}^n \cos(kx).
$
By substituting $n=(M-1)/2$ and squaring, we get
\begin{IEEEeqnarray}{rCl}\nonumber
\sin^2(Mx/2) &=&\sin^2(x/2) \left(1+2\sum_{k=1}^{(M-1)/2} \cos(kx) \right)^2\\
&\leq& \left(1+ 2 \cdot (M-1)/2\right)^2  \sin^2(x/2) = M^2 \sin^2(x/2)
\end{IEEEeqnarray}
with equality iff  $\cos(kx)=1$ for $k=1,2,\cdots, (M-1)/2$, i.e., $x\in \mathcal{S}$.

Assume that integer $M\geq 4$ is even.
By the double angle formula $\sin(2x)=2\sin(x)\cos(x),$ we get
\begin{IEEEeqnarray}{rCl}\nonumber
\sin(Mx/2) 
&=& M\sin(x) \prod_{k=1}^{M/2} \cos(kx/2).
\end{IEEEeqnarray}
Thus, $\sin^2(Mx/2) \leq M^2\sin^2(x)$ with equality iff $\cos^2(kx/2)=1$ for $k=1,\cdots, M/2$, i.e., $x\in \mathcal{S}$.
\end{proof}


\end{document}